\documentclass[submission,copyright,creativecommons]{eptcs}
\providecommand{\event}{GCM 2021} 
\usepackage{soul}

\usepackage{amsthm,amsmath,amssymb}

\theoremstyle{definition}
\newtheorem{theorem}{Theorem}[section]
\newtheorem{lemma}[theorem]{Lemma}
\newtheorem{proposition}[theorem]{Proposition}
\newtheorem{corollary}[theorem]{Corollary}

\newtheorem{example}[theorem]{Example}

\newenvironment{allintypewriter}{\ttfamily}{\par}

\usepackage{stmaryrd}

\newcommand{\dder}{\Rightarrow}

\newcommand{\mtt}{\ttt}
\newcommand{\mrm}{\mathrm}

\newcommand{\failrm}{\mathrm{fail}}
\newcommand{\failtt}{\ttt{fail}}
\newcommand{\skiptt}{\ttt{skip}}

\newcommand{\tuple}[1]{\langle#1\rangle}

\newcommand{\ifte}[3]{\ttt{if}\ #1\ \ttt{then}\ #2\ \ttt{else}\ #3}
\newcommand{\ift}[2]{\ttt{if}\ #1\ \ttt{then}\ #2}
\newcommand{\tryt}[2]{\ttt{try}\ #1\ \ttt{then}\ #2}
\newcommand{\trye}[2]{\ttt{try}\ #1\ \ttt{else}\ #2}
\newcommand{\tryte}[3]{\ttt{try}\ #1\ \ttt{then}\ #2\ \ttt{else}\ #3}

\newcommand{\ttt}{\texttt}

\usepackage{tabularx,environ}
\makeatletter
\newcommand{\specinput}[1]{\gdef\@specinput{#1}}%
\newcommand{\specoutput}[1]{\gdef\@specoutput{#1}}
\NewEnviron{spec}{
  \BODY
  \par\addvspace{.25\baselineskip}
  \noindent
  \begin{tabularx}{\textwidth}{@{\hspace{\parindent}} l X c}
    \textbf{Input:} & \@specinput \\
    \textbf{Output:} & \@specoutput
  \end{tabularx}
  \par\addvspace{.25\baselineskip}
}
\makeatother

\newcommand{\G}{\mathcal{G}}

\usepackage{cases}

\usepackage{caption}
\usepackage{subcaption}

\usepackage{breqn}

\usepackage{algorithm2e}

\usepackage{tikz}
\usetikzlibrary{cd}
\usetikzlibrary{shapes.misc}
\tikzset{>=latex}
\tikzset{every edge/.style={draw,thick}}
\usepackage{adjustbox}

\usepackage{xcolor}
\definecolor{bc-green-light}{RGB}{69, 191, 156}
\definecolor{bc-green-dark}{RGB}{37, 154, 85}
\definecolor{bc-blue-light}{RGB}{153, 187, 255}
\definecolor{bc-blue-dark}{RGB}{100, 99, 198}
\definecolor{bc-red-light}{RGB}{236, 107, 116}
\definecolor{bc-red-dark}{RGB}{161, 41, 41}
\definecolor{bc-pink-light}{RGB}{239, 161, 193}
\definecolor{bc-pink-dark}{RGB}{192, 62, 156}
\definecolor{bc-grey-light}{RGB}{196, 192, 200}
\definecolor{bc-grey-dark}{RGB}{135, 135, 135}

\tikzset{gp2 node/.style={draw, rounded rectangle, minimum height=.55cm, minimum size=6mm, thick}}
\tikzset{gp2 red node/.style={draw, rounded rectangle, minimum height=.55cm, minimum size=6mm, thick, fill=bc-red-light, TeXture={\tiny{\textcolor{bc-red-dark}{red}}}, TeXture sep=2pt}}
\tikzset{gp2 blue node/.style={draw, rounded rectangle, minimum height=.55cm, minimum size=6mm, thick, fill=bc-blue-light, TeXture={\tiny{\textcolor{bc-blue-dark}{blue}}}, TeXture sep=2pt}}
\tikzset{gp2 green node/.style={draw, rounded rectangle, minimum height=.55cm, minimum size=6mm, thick, fill=bc-green-light, fill=bc-green-light, TeXture={\tiny{\textcolor{bc-green-dark}{green}}}, TeXture sep=2pt}}
\tikzset{gp2 grey node/.style={draw, rounded rectangle, minimum height=.55cm, minimum size=6mm, thick, fill=bc-grey-light, TeXture={\tiny{\textcolor{bc-grey-dark}{grey}}}, TeXture sep=2pt}}
\tikzset{gp2 any node/.style={draw, rounded rectangle, minimum height=.55cm, minimum size=6mm, thick, fill=bc-pink-light, TeXture={\tiny{\textcolor{bc-pink-dark}{any}}}, TeXture sep=2pt}}

\tikzset{gp2 root/.style={draw, rounded rectangle, minimum height=.55cm, minimum size=6mm, thick, double, double distance=0.3mm}}
\tikzset{gp2 red root/.style={draw, rounded rectangle, minimum height=.55cm, minimum size=6mm, thick, double, double distance=0.3mm, fill=bc-red-light, TeXture={\tiny{\textcolor{bc-red-dark}{red}}}, TeXture sep=2pt}}
\tikzset{gp2 blue root/.style={draw, rounded rectangle, minimum height=.55cm, minimum size=6mm, thick, double, double distance=0.3mm, fill=bc-blue-light, TeXture={\tiny{\textcolor{bc-blue-dark}{blue}}}, TeXture sep=2pt}}
\tikzset{gp2 green root/.style={draw, rounded rectangle, minimum height=.55cm, minimum size=6mm, thick, double, double distance=0.3mm, fill=bc-green-light, TeXture={\tiny{\textcolor{bc-green-dark}{green}}}, TeXture sep=2pt}}
\tikzset{gp2 grey root/.style={draw, rounded rectangle, minimum height=.55cm, minimum size=6mm, thick, double, double distance=0.3mm, fill=bc-grey-light, TeXture={\tiny{\textcolor{bc-grey-dark}{grey}}}, TeXture sep=2pt}}
\tikzset{gp2 any root/.style={draw, rounded rectangle, minimum height=.55cm, minimum size=6mm, thick, double, double distance=0.3mm, fill=bc-pink-light, TeXture={\tiny{\textcolor{bc-pink-dark}{any}}}, TeXture sep=2pt}}

\tikzset{every edge quotes/.style={font=\tiny}}
\tikzset{gp2 edge/.style={->,thick,line width=1.1pt}}
\tikzset{gp2 red edge/.style={->,bc-red-light,"red"{bc-red-dark},sloped,line width=1.1pt}}
\tikzset{gp2 blue edge/.style={->,bc-blue-light,"blue"{bc-blue-dark},sloped,line width=1.1pt}}
\tikzset{gp2 green edge/.style={->,bc-green-light,"green"{bc-green-dark},sloped,line width=1.1pt}}
\tikzset{gp2 dashed edge/.style={->,thick,dashed,line width=1.1pt}}
\tikzset{gp2 any edge/.style={->,bc-pink-light,"any"{bc-pink-dark},sloped,line width=1.1pt}}

\tikzset{gp2 bi edge/.style={-,thick,line width=1.1pt}}
\tikzset{gp2 red bi edge/.style={-,line width=5pt,bc-red-light,"red"{bc-red-dark},sloped,line width=1.1pt}}
\tikzset{gp2 blue bi edge/.style={-,line width=5pt,bc-blue-light,"blue"{bc-blue-dark},sloped,line width=1.1pt}}
\tikzset{gp2 green bi edge/.style={-,line width=5pt,bc-green-light,"green"{bc-green-dark},sloped,line width=1.1pt}}
\tikzset{gp2 dashed bi edge/.style={-,thick,dashed,line width=1.1pt}}
\tikzset{gp2 any bi edge/.style={-,line width=5pt,bc-pink-light,"any"{bc-pink-dark},sloped,line width=1.1pt}}

\tikzset{none/.style={}}

\usepackage{pgfplots}
\pgfdeclarelayer{edgelayer}
\pgfdeclarelayer{nodelayer}
\pgfsetlayers{edgelayer,nodelayer,main}
\pgfplotsset{compat=1.16}

\newcommand*\tikzTeXture[1]{
  \begingroup
    \setbox0=\vbox spread \pgfkeysvalueof{/tikz/TeXture y sep}{\vfil%
               \hbox spread \pgfkeysvalueof{/tikz/TeXture x sep}{\hfil#1\hfil}\vfil}%
    \def\fillareasize{\pgfpointdiff 
      {\pgfpointanchor{path picture bounding box}{south west}}%
      {\pgfpointanchor{path picture bounding box}{north east}}}%
    \pgfextractx{\dimen0}{\fillareasize}
    \pgfextracty{\dimen2}{\fillareasize}
    \leavevmode\hbox to \dimexpr\dimen0+\wd0{%
      \cleaders\vbox to \dimexpr\dimen2+\ht0+\dp0{
        \cleaders\box0\vfil
      }\hfil
    }%
  \endgroup
}
\tikzset{TeXture/.style={path picture={
  \node[anchor=center,text width=,text height=]
        at (path picture bounding box.center) {\tikzTeXture{#1}};}
}}
\tikzset{TeXture/.default=\TeX}
\tikzset{TeXture x sep/.initial=0pt}
\tikzset{TeXture y sep/.initial=0pt}
\tikzset{TeXture sep/.style={/tikz/TeXture x sep=#1,/tikz/TeXture y sep=#1}}

\title{A Small-Step Operational Semantics for \texorpdfstring{GP\,2}{GP 2}}
\author{Brian Courtehoute and Detlef Plump
\institute{Department of Computer Science, University of York, York, UK}
\email{\{bc956,detlef.plump\}@york.ac.uk}
}
\def\titlerunning{A Small-Step Operational Semantics for \texorpdfstring{GP\,2}{GP 2}}
\def\authorrunning{B. Courtehoute \& D. Plump}
\begin{document}
\maketitle

\begin{abstract}
The operational semantics of a programming language is said to be small-step if each transition step is an atomic computation step in the language. A semantics with this property faithfully corresponds to the implementation of the language. The previous semantics of the graph programming language GP\,2 is not fully small-step because the loop and branching commands are defined in big-step style. In this paper, we present a truly small-step operational semantics for GP\,2 which, in particular, accurately models diverging computations. To obtain small-step definitions of all commands, we equip the transition relation with a stack of host graphs and associated operations. We prove that the new semantics is non-blocking in that every computation either diverges or eventually produces a result graph or the failure state. We also show the finite nondeterminism property, viz. that each configuration has only a finite number of direct successors. The previous semantics of GP\,2 is neither non-blocking nor does it have the finite nondeterminism property. We also show that, for a program and a graph that terminate, both semantics are equivalent, and that the old semantics can be simulated with the new one.
\end{abstract}

\section{Introduction}
\label{sec:intro}

\newcommand{\Ha}{
\tikz{
    \node (a) at (0,0) [draw,circle,line width=1pt,inner sep=.5ex,fill=bc-grey-light] {};
    \node (b) at (.6,0) [draw,circle,line width=1pt,inner sep=.5ex] {};
}}
\newcommand{\Hb}{
\tikz{
    \node (a) at (0,0) [draw,circle,line width=1pt,inner sep=.5ex,fill=bc-grey-light] {};
}}
\newcommand{\Hc}{
\tikz{
    \node (a) at (0,0) [draw,circle,line width=1pt,inner sep=.5ex] {};
}}

GP\,2 is a nondeterministic programming language based on graph transformation rules. The previous semantics of GP\,2 is defined by both small-step and big-step inference rules \cite{Plump17a}. An operational semantics is \emph{small-step} if atomic computation steps in the language correspond to transition steps, meaning that the language can be implemented by translating the transition steps into corresponding code. In this paper, we present a truly small-step operational semantics for GP\,2 which, in particular, accurately models diverging computations.

While the previous semantics (Figure \ref{fig:semantics-old}) has small-step elements, the branching and loop constructs are not small-step. This can lead to the semantic transition sequence \emph{blocking} or \emph{getting stuck} \cite{NielsonNielson07}, i.e.\ reaching a configuration which is neither a graph nor the failure state, such that no inference rule is applicable.

To illustrate this situation, consider the program \ttt{P = try (r1!)\phantom{o}then skip else skip}, with the rule \ttt{r1} : \mbox{$_1$\hspace{-.2em}\Hb $\,\,\,\Rightarrow$ $_1$\hspace{-.2em}\Ha}, applied to the host graph \Hb. The statement \ttt{r1!}\! means that the rule \ttt{r1} is called until it is no longer applicable. The \ttt{try} statement attempts to evaluate \ttt{r1!}\! but will neither branch to the \ttt{then} nor the \ttt{else} part because  the loop \ttt{r1!}\! diverges on \Hb. In the previous semantics, \ttt{try} statements are handled with the following inference rules :

\medskip
\begin{tabular}{ll}
$\mathrm{[try_1']}$ $\frac{\displaystyle \tuple{C,\, G} \rightsquigarrow^+ H}{\displaystyle \tuple{\tryte{C}{P}{Q},\, G}\rightsquigarrow \tuple{P,\, H}}$
&
$\mathrm{[try_2']}$ $\frac{\displaystyle \tuple{C,\, G} \rightsquigarrow^+ \failrm}{\displaystyle \tuple{\tryte{C}{P}{Q},\, G} \rightsquigarrow \tuple{Q,\, G}}$
\end{tabular}

\medskip
The premises of these inference rules are that the conditional part $C$ of a \ttt{try} statement applied to host graph $G$ results in either a graph $H$ or failure, which determines whether $P$ or $Q$ is called. If $\tuple{C,G}$ diverges (does not terminate) however, neither rule applies. Since there are no other \ttt{try} rules, the transition sequence gets stuck. 

The new semantics we introduce in this paper handles \ttt{try} statements with the following rules:

\medskip
\begin{tabular}{ll}
\multicolumn{2}{l}{$\mathrm{[try_1]}$ $\tuple{\tryte{C}{P}{Q},\, S} \to \tuple{\text{TRY}(C,P,Q),\, \text{push}(\text{top}(S),\,S)}$}
\end{tabular}

\medskip
\begin{tabular}{ll}
$\mathrm{[try_2]}$ $\frac{\displaystyle \tuple{C,\, S} \to \tuple{C',\, S'}}{\displaystyle \tuple{\text{TRY}(C,P,Q),\, S}\to \tuple{\text{TRY}(C',P,Q),\, S'}}$
&
$\mathrm{[try_3]}$ $\frac{\displaystyle \tuple{C,\, S} \to S'}{\displaystyle \tuple{\text{TRY}(C,P,Q),\, S}\to \tuple{P,\, \text{pop2}(S')}}$
\end{tabular}

\medskip
\begin{tabular}{ll}
$\mathrm{[try_4]}$ $\frac{\displaystyle \tuple{C,\, S} \to \failrm}{\displaystyle \tuple{\text{TRY}(C,P,Q),\, S} \to \tuple{Q,\, \text{pop}(S)}}$
&
\end{tabular}

\medskip
Here $S$ and $S'$ are stacks of graphs. The rule $\mathrm{[try_1]}$ duplicates the top of the stack, and the TRY construct signals that the copy operation has happened. Repeated applications of the inference rule $\mathrm{[try_2]}$ model the evaluation of the condition in a small-step fashion. If the condition loops, $\mathrm{[try_2]}$ can be applied indefinitely, and we get an infinite transition sequence.

Intuitively, \texttt{P} should loop, which is what happens in the implementation of GP\,2. In the previous semantics however, \ttt{P} gets stuck because \ttt{r1!}\/ diverges, which means that we cannot apply either of the inference rules $\mathrm{[try_1']}$ or $\mathrm{[try_2']}$ to resolve the \ttt{try} statement.

The previous semantics tries to remedy this issue in the \emph{semantic function} which associates to a program \ttt{P} and host graph $G$ the set $[ \ttt{P} ] G$ of all possible outcomes of the execution of \ttt{P} on $G$. These outcomes can be a graph, the element fail, or $\bot$ which represents an infinite transition sequence. The previous semantic function uses $\bot$ as an outcome if the transition sequence gets stuck. However, there are problems with this approach.

Consider the program \ttt{P = try Loop then skip else skip}, where \ttt{Loop = \{r1,r2\}!}, \ttt{r1} is as previously defined, and \ttt{r2} : \mbox{\Hb $\,\,\,\Rightarrow$ $\emptyset$}. The command \ttt{\{r1,r2\}} is a \emph{rule set call}, meaning that rules \ttt{r1} and \ttt{r2} are selected nondeterministically. When \ttt{P} is executed on the host graph \Hb\,, an application of \ttt{r2} causes the loop to terminate since it removes the marked node which is necessary for either rule to be applicable. Hence \ttt{r1} may be applied a number of times, and then \ttt{r2} is applied once. But it should also be possible that \ttt{r2} is never called, resulting in a diverging computation. Hence the set of outcomes we want is $\{\bot,\, \emptyset,\, \Hc,\, \Hc\Hc,\, \Hc\Hc\Hc,\, \dots\}$. According to the previous semantics, however, the execution of \ttt{P} on \Hb\, cannot get stuck since \ttt{Loop} \emph{can}\/ always transition to a graph; and by the rules $\mathrm{[try_1']}$ and $\mathrm{[try_2']}$, the execution cannot diverge either. So $\bot \not\in [ \ttt{P} ] \Hb = \{\emptyset,\, \Hc,\, \Hc\Hc,\, \Hc\Hc\Hc,\, \dots\}$. Now, the new semantics indeed corresponds exactly to our intuition of the operational behaviour of GP\,2 programs. Moreover, we conjecture that the implementation is sound with respect to the new semantics, in that the behaviour of the implementation is covered by the new semantics.

This may also lead to two programs being semantically equivalent, even though they should not be. Programs $P$ and $P'$ are \emph{semantically equivalent} if $[ P ] = [ P' ]$, i.e. they have the same outcomes for all host graphs. Consider the program \ttt{P = try (\{r3,r2\}!)\phantom{o}then skip else skip}, where \ttt{r3} : \mbox{$_1$\hspace{-.2em}\Hb $\,\,\,\Rightarrow$ $_1$\hspace{-.2em}\Hb}. It can diverge but is semantically equivalent to \ttt{Q = try r2 then skip else skip} since the previous semantics cannot detect that divergence. For instance, $[ \ttt{P} ] \Hb = [ \ttt{Q}] \Hb = \{\emptyset\}$, but $[ \ttt{P} ] \Hb$ should include $\bot$.

The aforementioned issues can also happen with \ttt{if} statements, which work similarly to \ttt{try} statements, except that the changes the condition made to the host graph are reversed, even if the evaluation of the condition succeeds. Nested loops such as \ttt{Loop!} can get stuck as well since their inference rules also assume that the loop body either results in a graph or fails.

Diverging computations not being modelled properly entails an undesirable property, namely \emph{infinite nondeterminism}, i.e. there can be infinitely many configurations reachable in a single transition step. Consider the program \ttt{P = try Loop then skip else skip}, where \ttt{Loop = \{r1,r2\}!}, and the rules are as previously defined. We have $[ \ttt{Loop} ] \Hb = \{\emptyset,\, \Hc,\, \Hc\Hc,\, \Hc\Hc\Hc,\, \dots\, ,\, \bot\}$. In a transition sequence starting with $\tuple{\ttt{P},\Hb}$, since the \ttt{try} statement is resolved within a single step, it only takes one step to transition to either of the graphs in the set $\{\emptyset,\, \Hc,\, \Hc\Hc,\, \Hc\Hc\Hc,\, \dots\,\}$, of which there are infinitely many.

The semantics we introduce in this paper is truly small-step and as such, it accurately models looping computations with diverging transition sequences. When starting with a valid GP\,2 program, it cannot get stuck, which is a property we call \emph{non-blocking}. As a consequence of the small-step approach, we get \emph{finite nondeterminism}, meaning we can only reach a finite number of configurations within a single transition step.

In Section \ref{sec:gp2}, we give a brief overview of the rule-based graph programming language GP\,2 along with the previous semantics. We propose the new semantics in Section \ref{sec:semantics} and give examples of transition sequences. In Section \ref{sec:semantic-props} we prove several properties of the new semantics, including non-blocking as well as finite nondeterminism, and define the semantic function along with semantic equivalence. In Section \ref{sec:comparison}, we compare the new and previous semantics by showing the new semantic function is an extension of the previous one, and that they are equivalent excluding divergence.

\section{The Graph Programming Language \texorpdfstring{GP\,2}{GP 2}}
\label{sec:gp2}

This section provides a brief introduction to GP\,2 \cite{Plump12a}, a nondeterministic graph programming language based on transformation rules. We show the abstract syntax of GP\,2 programs below, and refer to \cite{Bak15a} for the full syntax. The language is implemented by a compiler generating C code \cite{Bak-Plump16a}.

GP\,2 programs transform input graphs into output graphs, where graphs are labelled and directed and may contain parallel edges and loops.

The principal programming construct in GP\,2 are conditional graph transformation rules labelled with expressions. For example, Figure \ref{fig:is-cyclic} shows a program recognising graphs that contain cycles and the declaration of its rules. The rule \ttt{delete} which has three formal parameters, a left-hand graph and a right-hand graph which are specified graphically, and a textual condition starting with the keyword \ttt{where}. The small numbers attached to nodes are identifiers, all other text in the graphs are labels.

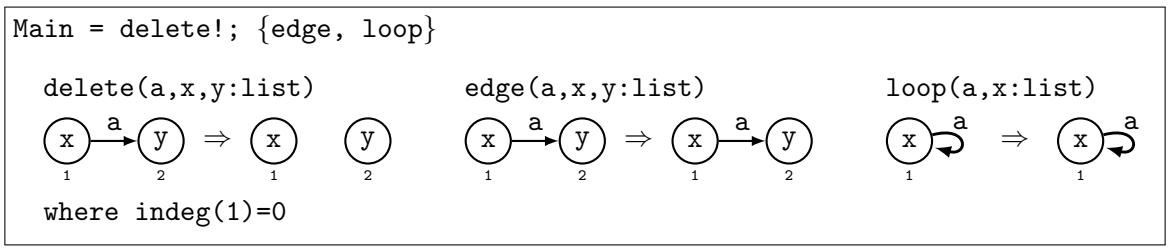
\begin{figure}[!ht]
\centering
\fbox{\begin{minipage}{.95\textwidth}
\begin{allintypewriter}
Main = delete!; \{edge, loop\}

\setlength{\tabcolsep}{0.4cm}

\medskip
\smallskip

\begin{tabular}{ p{4.8cm} p{4.8cm} p{2cm} }
	
	delete(a,x,y:list) & edge(a,x,y:list) & loop(a,x:list)\\
	
	\begin{tikzpicture}
        
		\node (a) at (0,0)       [gp2 node] {x};
		\node (b) at (1.25,0)    [gp2 node] {y};
		
		\node (d) at (2,0)   {$\Rightarrow$};
		
		\node (e) at (2.75,0)    [gp2 node] {x};
		\node (f) at (4,0)     [gp2 node] {y};
		
		\node (A) at (0,-.45)    {\tiny{1}};
		\node (B) at (1.25,-.45) {\tiny{2}};
		\node (E) at (2.75,-.45) {\tiny{1}};
		\node (F) at (4,-.45)  {\tiny{2}};
		
		\draw (a) edge[->,thick] node[above] {a} (b);
	\end{tikzpicture}
	&
	
	\begin{tikzpicture}
        
		\node (a) at (0,0)       [gp2 node] {x};
		\node (b) at (1.25,0)    [gp2 node] {y};
		
		\node (d) at (2,0)   {$\Rightarrow$};
		
		\node (e) at (2.75,0)    [gp2 node] {x};
		\node (f) at (4,0)     [gp2 node] {y};
		
		\node (A) at (0,-.45)    {\tiny{1}};
		\node (B) at (1.25,-.45) {\tiny{2}};
		\node (E) at (2.75,-.45) {\tiny{1}};
		\node (F) at (4,-.45)  {\tiny{2}};
		
		\draw (a) edge[->,thick] node[above] {a} (b)
		      (e) edge[->,thick] node[above] {a} (f);
	\end{tikzpicture}
	&
	\begin{tikzpicture}
        
		\node (a) at (0,0)       [gp2 node] {x};
		
		\node (d) at (1.4,0)   {$\Rightarrow$};
		
		\node (e) at (2.275,0)    [gp2 node] {x};
		
		\node (A) at (0,-.45)    {\tiny{1}};
		\node (E) at (2.275,-.45) {\tiny{1}};
		
		\draw (a) edge[in=-15,out=15,loop,gp2 edge] node[above] {a} (a)
		      (e) edge[in=-15,out=15,loop,gp2 edge] node[above] {a} (e);
	\end{tikzpicture}
	\\
	\vspace{-1em}where indeg(1)=0 & &
\end{tabular}

\end{allintypewriter}
\end{minipage}}
\caption{GP\,2 program recognising cyclic graphs}
\label{fig:is-cyclic}
\end{figure}

\bigskip
GP\,2 labels consist of an expression and an optional mark (explained below). Expressions are of type \texttt{int}, \texttt{char}, \texttt{string}, \texttt{atom} or \texttt{list}, where \texttt{atom} is the union of \texttt{int} and \texttt{string}, and \texttt{list} is the type of a (possibly empty) list of atoms. Lists of length one are equated with their entries and hence every expression can be considered as a list.

The concatenation of two lists $x$ and $y$ is written $x{:}y$, the empty list is denoted by \ttt{empty}. Character strings are enclosed in double quotes. Composite arithmetic expressions such as $\ttt{n * n}$ must not occur in the left-hand graph, and all variables occurring in the right-hand graph or the condition must also occur in the left-hand graph. 

Besides carrying list expressions, nodes and edges can be \emph{marked}. For example, one of the nodes in rule r1 in the introduction is marked by a grey shading. Marks are convenient to highlight items in input or output graphs, and to record visited items during a graph traversal. For instance, a graph can be checked for connectedness by propagating marks along edges as long as possible and subsequently testing whether any unmarked nodes remain. Note that conventional graph algorithms are often described by using marks as a visual aid \cite{CormenLeisersonRivestStein09}.

Additionally, nodes in rules and host graphs can be \emph{rooted}. If such a node appears in the left-hand side of a rule, it can only be matched with a root node in the host graph. Their use restricts matching to the neighbourhoods of root nodes, which can greatly increase efficiency \cite{Campbell-Courtehoute-Plump20a}.

We do not elaborate any further on features such as marks or roots because the GP\,2 semantics does not depend on them.

\bigskip
Rules operate on \emph{host graphs}\/ which are labelled with constant values (lists containing integer and string constants). Applying a rule $L \dder R$\/ to a host graph $G$\/ works roughly as follows: (1) Replace the variables in $L$ and $R$\/ with constant values and evaluate the expressions in $L$ and $R$, to obtain an instantiated rule $\hat{L} \dder \hat{R}$. (2) Choose a subgraph $S$\/ of $G$\/ isomorphic to $\hat{L}$ such that the dangling condition and the rule's application condition are satisfied (see below). (3) Replace $S$\/ with $\hat{R}$\/ as follows: numbered nodes stay in place (possibly relabelled), edges and unnumbered nodes of $\hat{L}$ are deleted, and edges and unnumbered nodes of $\hat{R}$ are inserted. 

In this construction, the \emph{dangling condition}\/ requires that nodes in $S$\/ corresponding to unnumbered nodes in $\hat{L}$ (which should be deleted) must not be incident with edges outside $S$. The rule's application condition is evaluated after variables have been replaced with the corresponding values of $\hat{L}$, and node identifiers of $L$\/ with the corresponding identifiers of $S$. For example, the term $\ttt{indeg(1)=0}$ in the condition of \ttt{delete} in Figure \ref{fig:is-cyclic} forbids the node $g(\ttt{1})$ to have incoming edges, where $g(\ttt{1})$ is the node in $S$ corresponding to \ttt{1}.

\begin{figure}

\begin{tabular}{lll}
Program         & $::=$ & Declaration \{ Declaration \} \\
Declaration     & $::=$ & MainDecl  \;$\vert$\;  ProcedureDecl  \;$\vert$\;  RuleDecl \\
MainDecl        & $::=$ & \ttt{Main} `$=$' CommandSeq \\
ProcedureDecl   & $::=$ & ProcedureID `$=$' [ `[' LocalDecl `]' ] CommandSeq \\
LocalDecl       & $::=$ & ( RuleDecl \;$\vert$\; ProcedureDecl ) \{ LocalDecl \} \\
CommandSeq      & $::=$ & Command \{`;'Command\} \\
Command         & $::=$ & Block \\
                &       & $\vert$\; \ttt{if} Block \ttt{then} Block [ \ttt{else} Block ]\\
                &       & $\vert$\; \ttt{try} Block [ \ttt{then} Block ] [ \ttt{else} Block ] \\
Block           & $::=$ & `(' CommandSeq `)' [`!'] \\
                &       & $\vert$\; SimpleCommand \\
                &       & $\vert$\; Block \ttt{or} Block \\
SimpleCommand   & $::=$ & RuleSetCall [`!'] \\
                &       & $\vert$\; ProcedureCall [`!'] \\
                &       & $\vert$\; \ttt{break} \\
                &       & $\vert$\; \ttt{skip} \\
                &       & $\vert$\; \ttt{fail} \\
RuleSetCall     & $::=$ & RuleID \;$\vert$\; `\{' [ RuleID \{ `,' RuleID \} ] `\}' \\
ProcedureCall   & $::=$ & ProcedureID
\end{tabular}

\caption{GP\,2 Program Syntax}
\label{fig:syntax-program}
\end{figure}

\begin{figure}[p]
\centering

\begin{subfigure}{\textwidth}
\begin{center}
\begin{tabular}{ll}
$\mathrm{[call_1']}$ $\frac{\displaystyle G \dder_R H}{\displaystyle\tuple{R,\,G} \rightsquigarrow H}$ 
&
$\mathrm{[call_2']}$ $\frac{\displaystyle G \not\dder_R}{\displaystyle\tuple{R,\,G} \rightsquigarrow \failrm}$
\\\\
$\mathrm{[seq_1']}$ $\frac{\displaystyle \tuple{P,\, G} \rightsquigarrow \tuple{P',\, H}}{\displaystyle \tuple{P;Q,\, G} \rightsquigarrow \tuple{P';Q,\, H}}$ 
&
$\mathrm{[seq_2']}$ $\frac{\displaystyle \tuple{P,\, G} \rightsquigarrow H}{\displaystyle \tuple{P;Q,\, G}\rightsquigarrow \tuple{Q,\, H}}$
\\\\
$\mathrm{[seq_3']}$ $\frac{\displaystyle \tuple{P,\, G} \rightsquigarrow \failrm}{\displaystyle \tuple{P;Q,\, G}\rightsquigarrow \failrm}$
\\\\
$\mathrm{[if_1']}$ $\frac{\displaystyle \tuple{C,\, G} \rightsquigarrow^+ H}{\displaystyle \tuple{\ifte{C}{P}{Q},\, G}\rightsquigarrow \tuple{P,\, G}}$
&
$\mathrm{[if_2']}$ $\frac{\displaystyle \tuple{C,\, G} \rightsquigarrow^+ \failrm}{\displaystyle \tuple{\ifte{C}{P}{Q},\, G} \rightsquigarrow \tuple{Q,\, G}}$
\\\\
$\mathrm{[try_1']}$ $\frac{\displaystyle \tuple{C,\, G} \rightsquigarrow^+ H}{\displaystyle \tuple{\tryte{C}{P}{Q},\, G}\rightsquigarrow \tuple{P,\, H}}$
&
$\mathrm{[try_2']}$ $\frac{\displaystyle \tuple{C,\, G} \rightsquigarrow^+ \failrm}{\displaystyle \tuple{\tryte{C}{P}{Q},\, G} \rightsquigarrow \tuple{Q,\, G}}$
\\\\
$\mathrm{[alap_1']}$ $\frac{\displaystyle \tuple{P,\, G} \rightsquigarrow^+ H}{\displaystyle \tuple{P!,\, G} \rightsquigarrow \tuple{P!,\, H}}$
&
$\mathrm{[alap_2']}$ $\frac{\displaystyle \tuple{P,\, G} \rightsquigarrow^+ \failrm}{\displaystyle \tuple{P!,\, G} \rightsquigarrow G}$
\\\\
$\mrm{[alap_3']}$ $\frac{\displaystyle \tuple{P,\, G} \rightsquigarrow^* \tuple{\mtt{break}, H}}{\displaystyle \tuple{P!,\, G} \rightsquigarrow H}$
&
$\mrm{[break']}$ $\tuple{\mtt{break}; P,\, G} \rightsquigarrow \tuple{\mtt{break},\, G}$
\end{tabular}
\end{center}
\vspace{-0.333333em}
\caption{Inference rules for core commands}
\label{fig:semantics-old-core}
\end{subfigure}

\vspace{3em}
\begin{subfigure}{\textwidth}
\begin{center}
\begin{tabular}{llll}
$\mrm{[or_1']}$ & $\tuple{P\mspace{.5mu} \mathop{\mtt{or}}\, Q,\, G} \rightsquigarrow \tuple{P,\, G}$ \hspace{4em} & $\mrm{[or_2']}$ & $\tuple{P\mspace{.5mu} \mathop{\mtt{or}}\, Q,\, G} \rightsquigarrow \tuple{Q,\, G}$
\\[1.5ex]
$\mrm{[skip']}$ & $\tuple{\skiptt,\, G} \rightsquigarrow G$ & $\mrm{[fail']}$ & $\tuple{\failtt,\, G} \rightsquigarrow \failrm$
\\[1.5ex]
$\mrm{[if_3']}$ & \multicolumn{3}{l}{$\tuple{\ift{C}{P},\, G} \rightsquigarrow \tuple{\ifte{C}{P}{\skiptt},\, G}$}
\\[1.5ex]
$\mrm{[try_3']}$ & \multicolumn{3}{l}{$\tuple{\tryt{C}{P},\, G} \rightsquigarrow \tuple{\tryte{C}{P}{\skiptt},\, G}$}
\\[1.5ex]
$\mathrm{[try_4']}$ & \multicolumn{3}{l}{$\tuple{\trye{C}{P},\, G} \rightsquigarrow \tuple{\tryte{C}{\skiptt}{P},\, G}$}
\\[1.5ex]
$\mathrm{[try_5']}$ & \multicolumn{3}{l}{$\tuple{\mtt{try}\ C,\, G} \rightsquigarrow \tuple{\tryte{C}{\skiptt}{\skiptt},\, G}$}
\end{tabular}
\end{center}
\vspace{-0.333333em}
\caption{Inference rules for derived commands}
\label{subfig:semantics-old-derived}
\end{subfigure}

\vspace{2em}
\caption{Previous GP\,2 Semantics}
\label{fig:semantics-old}
\end{figure}

\bigskip
Formally, GP\,2 is based on a form of attributed graph transformation according to the so-called double-pushout approach \cite{Hristakiev-Plump16a,Ehrig-Ehrig-Prange-Taentzer06a}. The grammar in Figure \ref{fig:syntax-program} gives the abstract syntax of GP\,2 programs. A program consists of declarations of conditional rules and procedures, and exactly one declaration of a main command sequence. The category RuleID refers to declarations of conditional rules in RuleDecl (whose syntax is omitted). Procedures must be non-recursive, they can be seen as macros with local declarations.

The call of a rule set $\{r_1,\dots,r_n\}$ nondeterministically applies one of the rules whose left-hand graph matches a subgraph of the host graph such that the dangling condition and the rule's application condition are satisfied. The call \emph{fails}\/ if none of the rules is applicable to the host graph. 

The command \ttt{if} $C$ \ttt{then} $P$ \ttt{else} $Q$ is executed on a host graph $G$ by first executing $C$ on $G$. If this results in a graph, $P$\/ is executed on the original graph $G$; otherwise, if $C$ fails, $Q$ is executed on $G$. The \ttt{try} command has a similar effect, except that $P$\/ is executed on the result of $C$'s execution in case $C$ succeeds. 

The loop command $P!$ executes the body $P$\/ repeatedly until it fails. When this is the case, $P!$ terminates with the graph on which the body was entered for the last time. The \ttt{break} command inside a loop terminates that loop with the current graph and transfers control to the command following the loop.

A program $P$ \ttt{or} $Q$ non-deterministically chooses to execute either $P$ or $Q$, which can be simulated by a rule-set call and the other commands \cite{Plump12a}. The commands \ttt{skip} and \texttt{fail} can also be expressed by the other commands: \ttt{skip} is equivalent to an application of the rule $\emptyset \dder \emptyset$ (where $\emptyset$ is the empty graph) and \ttt{fail} is equivalent to an application of $\{\}$ (the empty rule set).


\bigskip
Like Plotkin's structural operational semantics \cite{Plotkin04a}, the previous GP\,2 semantics is given by inference rules. The rules in Figure \ref{fig:semantics-old} define the transition relation $\rightsquigarrow$ over the following set:
$$(\text{ComSeq} \times \mathcal{G}) \,\times\, ((\text{ComSeq} \times \mathcal{G}) \,\cup\, \mathcal{G} \,\cup\, \{\text{fail}\})\text{.}$$

Here $\mathcal{G}$ is the set of all GP\,2 host graphs and ComSeq is the set of command sequences as defined in the syntax (Figure \ref{fig:syntax-program}), where we assume that procedure IDs have been eliminated by macro expansion. This means that procedure IDs have been replaced with their defining command sequence, and name clashes arising from local declarations have been resolved by renaming.

The element fail represents the program resulting in a failure state. The inference rules contain universally quantified variables, namely host graphs $G$ and $H$, command sequences in ComSeq $C$, $P$, $P'$, and $Q$, and rule set call $R$. The transitive closure of $\rightsquigarrow$ is denoted by $\rightsquigarrow^+$, and the reflexive transitive closure by $\rightsquigarrow^*$.

In general, the execution of a program on a host graph may result in another graph, fail, or diverge. Also, executions can get \emph{stuck} in that they reach a non-terminal configuration (neither a graph nor fail) to which no inference rule is applicable. Let $\G$ be the set of all host graphs and $\G^{\oplus} = \G \cup \{\bot,\mrm{fail}\}$. These outcomes are described by the semantic function $[\_] \,:\, \text{ComSeq} \to (\mathcal{G} \to 2^{\mathcal{G}^{\oplus}})$ which, for a command sequence $P$ and a host graph $G$, is defined as
$$[ P ] G = \{X \in \mathcal{G} \cup \{\text{fail}\} \,|\, \tuple{P,G} \rightsquigarrow^+ X\} \,\cup\, \{\bot \,|\, P \text{ can diverge or get stuck from } G\}\text{.}$$

By divergence we mean non-termination, that is the existence of an infinite transition sequence starting in $\tuple{P,G}$.

\section{The Small-Step Semantics}
\label{sec:semantics}

In this section, we introduce an improved semantics defined by inference rules, and give examples of transition sequences.

Due to additional constructs, the new semantics needs to distinguish between command sequences that are valid GP\,2 programs, and command sequences that are intermediary. The former are members of CommandSeq from the syntax in Figure \ref{fig:syntax-program}, and are called \emph{command sequences}. They have to satisfy the context conditions specified in Appendix A.6 of Bak's thesis \cite{Bak15a}. The following condition is particularly relevant to this paper: ``A \ttt{break} must be enclosed within a loop. If a \ttt{break} is in the condition of a branching statement\footnote{By branching statement, we mean an \texttt{if}, \texttt{try}, ITE, or TRY statement.}, the enclosing loop must be within the same condition.'' This constraint is not specific to graph programs: Java, C, and Python have similar restrictions on the use of \ttt{break} statements.

We define \emph{extended command sequences} (set ExtComSeq) to be command sequences with additional auxiliary constructs ITE and TRY. They do not follow context conditions since we may want a \ttt{break} outside of a loop in an intermediary transition step. The ITE and TRY statements serve to advance the command sequence in the condition in a small-step fashion, as well as to maintain the stack of host graphs. When we enter an ITE or TRY statement, the top of the stack (and current host graph) is duplicated in order to keep a backup. When exiting these statements we either pop the top, modified graph, or the second graph on the stack which is the unmodified backup copy depending on the outcome of the condition. The stack structure is needed because \ttt{if} and \ttt{try} statements may be nested. Whenever we enter an ITE or TRY construct, we push a graph, and whenever we exit one, we pop a graph. This ensures that the stack always contains enough graphs to pop and that the current host graph is always on top.

\begin{figure}[p]
\centering

\begin{subfigure}{\textwidth}
\begin{center}
\begin{tabular}{ll}
$\mathrm{[call_1]}$ $\frac{\displaystyle \text{top}(S) \dder_R G}{\displaystyle\tuple{R,\,S} \to \text{push}(G,\text{pop}(S))}$ 
&
$\mathrm{[call_2]}$ $\frac{\displaystyle \text{top}(S) \not\dder_R}{\displaystyle\tuple{R,\,S} \to \failrm}$
\\\\
$\mathrm{[seq_1]}$ $\frac{\displaystyle \tuple{P,\, S} \to \tuple{P',\, S'}}{\displaystyle \tuple{P;Q,\, S} \to \tuple{P';Q,\, S'}}$ 
&
$\mathrm{[seq_2]}$ $\frac{\displaystyle \tuple{P,\, S} \to S'}{\displaystyle \tuple{P;Q,\, S}\to \tuple{Q,\, S'}}$
\\\\
$\mathrm{[seq_3]}$ $\frac{\displaystyle \tuple{P,\, S} \to \failrm}{\displaystyle \tuple{P;Q,\, S}\to \failrm}$
&
$\mrm{[break]}$ $\tuple{\mtt{break}; P,\, S} \to \tuple{\mtt{break},\, S}$
\\\\
$\mathrm{[alap_1]}$ $\tuple{P!,\, S} \to \tuple{ \tryte{P}{P!}{\skiptt},\, S}$
&
$\mrm{[alap_2]}$ $\tuple{\text{TRY}(\mtt{break},\,P!,\,\skiptt),\, S} \to \text{pop2}(S)$
\\\\
\multicolumn{2}{l}{$\mathrm{[if_1]}$ $\tuple{\ifte{C}{P}{Q},\, S} \to \tuple{\text{ITE}(C,P,Q),\, \text{push}(\text{top}(S),\,S)}$}
\\\\
\multicolumn{2}{l}{$\mathrm{[try_1]}$ $\tuple{\tryte{C}{P}{Q},\, S} \to \tuple{\text{TRY}(C,P,Q),\, \text{push}(\text{top}(S),\,S)}$}
\\\\
$\mathrm{[if_2]}$ $\frac{\displaystyle \tuple{C,\, S} \to \tuple{C',\, S'}}{\displaystyle \tuple{\text{ITE}(C,P,Q),\, S}\to \tuple{\text{ITE}(C',P,Q),\, S'}}$
&
$\mathrm{[try_2]}$ $\frac{\displaystyle \tuple{C,\, S} \to \tuple{C',\, S'}}{\displaystyle \tuple{\text{TRY}(C,P,Q),\, S}\to \tuple{\text{TRY}(C',P,Q),\, S'}}$
\\\\
$\mathrm{[if_3]}$ $\frac{\displaystyle \tuple{C,\, S} \to S'}{\displaystyle \tuple{\text{ITE}(C,P,Q),\, S}\to \tuple{P,\, \text{pop}(S')}}$
&
$\mathrm{[try_3]}$ $\frac{\displaystyle \tuple{C,\, S} \to S'}{\displaystyle \tuple{\text{TRY}(C,P,Q),\, S}\to \tuple{P,\, \text{pop2}(S')}}$
\\\\
$\mathrm{[if_4]}$ $\frac{\displaystyle \tuple{C,\, S} \to \failrm}{\displaystyle \tuple{\text{ITE}(C,P,Q),\, S} \to \tuple{Q,\, \text{pop}(S)}}$
&
$\mathrm{[try_4]}$ $\frac{\displaystyle \tuple{C,\, S} \to \failrm}{\displaystyle \tuple{\text{TRY}(C,P,Q),\, S} \to \tuple{Q,\, \text{pop}(S)}}$
\end{tabular}
\end{center}
\vspace{-0.333333em}
\caption{Inference rules for core commands}
\label{fig:semantics-new-core}
\end{subfigure}

\vspace{3em}
\begin{subfigure}{\textwidth}
\begin{center}
\begin{tabular}{llll}
$\mrm{[or_1]}$ & $\tuple{P\mspace{.5mu} \mathop{\mtt{or}}\, Q,\, S} \to \tuple{P,\, S}$ \hspace{4em} & $\mrm{[or_2]}$ & $\tuple{P\mspace{.5mu} \mathop{\mtt{or}}\, Q,\, S} \to \tuple{Q,\, S}$
\\[1.5ex]
$\mrm{[skip]}$ & $\tuple{\skiptt,\, S} \to S$ & $\mrm{[fail]}$ & $\tuple{\failtt,\, S} \to \failrm$
\\[1.5ex]
$\mrm{[if_5]}$ & \multicolumn{3}{l}{$\tuple{\ift{C}{P},\, S} \to \tuple{\ifte{C}{P}{\skiptt},\, S}$}
\\[1.5ex]
$\mrm{[try_5]}$ & \multicolumn{3}{l}{$\tuple{\tryt{C}{P},\, S} \to \tuple{\tryte{C}{P}{\skiptt},\, S}$}
\\[1.5ex]
$\mathrm{[try_6]}$ & \multicolumn{3}{l}{$\tuple{\trye{C}{P},\, S} \to \tuple{\tryte{C}{\skiptt}{P},\, S}$}
\\[1.5ex]
$\mathrm{[try_7]}$ & \multicolumn{3}{l}{$\tuple{\mtt{try}\ C,\, S} \to \tuple{\tryte{C}{\skiptt}{\skiptt},\, S}$}
\end{tabular}
\end{center}
\vspace{-0.333333em}
\caption{Inference rules for derived commands}
\label{subfig:semantics-new-derived}
\end{subfigure}

\vspace{2em}
\caption{Improved GP\,2 Semantics}
\label{fig:semantics-new}
\end{figure}

The rules in Figure \ref{fig:semantics-new} inductively define a transition relation $\to$ over the following set:
$$(\text{ExtComSeq} \times \mathcal{S}) \,\times\, ((\text{ExtComSeq} \times \mathcal{S}) \,\cup\, \mathcal{S} \,\cup\, \{\text{fail}\})\text{,}$$
where $\mathcal{S}$ is the set of all stacks of GP\,2 host graphs (explained below). We call an element of the set $(\text{ExtComSeq} \times \mathcal{S}) \,\cup\, \mathcal{S} \,\cup\, \{\text{fail}\}$ an \emph{extended configuration}, whereas $(\text{CommandSeq} \times \mathcal{S}) \,\cup\, \mathcal{S} \,\cup\, \{\text{fail}\}$ is the set of \emph{configurations}. A configuration (or extended configuration) $C$ is \emph{terminal} if $C=\text{fail}$ or $C=S$ for some graph stack $S$.

The set $\mathcal{S}$ is the set of all non-empty stacks of GP\,2 host graphs where the top element is the current host graph, and where the other elements are backup copies to revert to or discard after the resolution of conditions of branching statements. Such a stack $S=[G_1,G_2,G_3,\dots,G_n]$ is a finite ordered list of GP\,2 host graphs with unary operations $\text{top}(S)=G_1$, $\text{pop}(S)=[G_2,G_3,\dots,G_n]$ and $\text{pop2}(S)=[G_1,G_3,\dots,G_n]$, as well as the binary operation $\text{push}(G,S)=[G,G_1,G_2,\dots,G_n]$, where $G$ is a GP\,2 host graph.

\bigskip
Most of the inference rules in Figure \ref{fig:semantics-new} have a horizontal bar. These rules consist of a \emph{premise} above the bar and a \emph{conclusion} below. The conclusion defines a transition step provided that the premise holds. A rule without a bar is called an \emph{axiom} and can be applied to a configuration without any precondition.

There are several universally quantified meta-variables within the inference rules. $P$, $P'$, $Q$, $Q'$, $C$, and $C'$ stand for extended command sequences in ExtComSeq, $S$ stands for a graph stack in $\mathcal{S}$, $G$ represents a host graph, and $R$ represents a rule set. We denote the transitive closure of $\to$ by $\to^+$, and the reflexive transitive closure by $\to^*$.

The inference rules inductively define the transition relation $\to$. The rules $\mathrm{[call_1]}$ and $\mathrm{[call_2]}$ are base cases. Their premises are GP\,2 derivations. Which of the two premises is satisfied depends on whether $\text{top}(S) \dder_R G$ or $\text{top}(S) \not\dder_R$, i.e. whether a rule in the rule set can be applied to the current host graph or not. The \ttt{if} and \ttt{try} statements are modelled by the $\mathrm{[if_i]}$ and $\mathrm{[try_i]}$ rules.

Sequential composition of commands is covered by $\mathrm{[seq_1]}$, $\mathrm{[seq_2]}$, and $\mathrm{[seq_3]}$, covering the cases of whether the first command called on a host graph results in a configuration, a graph stack, or fail.

Loops are semantically described as a try statement in $\mathrm{[alap_1]}$. Calling a command sequence as long as possible is modelled by trying to apply the command sequence, and if it succeeds, keep applying it as long as possible. Breaking from a loop is handled by $\mathrm{[break]}$, which makes sure commands following the break are discarded, and $\mathrm{[alap_2]}$, which terminates the loop if there is an isolated break in the TRY condition.

Figure \ref{fig:semantics-new-core} shows the inference rules for the core commands of GP\,2, while Figure \ref{subfig:semantics-new-derived} gives the inference rules for derived commands such as \ttt{or}, \ttt{skip}, and \ttt{fail}, as well as some \ttt{if} and \ttt{try} statements with omitted \ttt{then} and \ttt{else} clauses. These commands are referred to as \emph{derived} commands because they can be defined by the core commands (see \cite{Plump12a} for the case of the previous semantics).

\bigskip
Let us look at a couple of examples of transition sequences in Figure \ref{fig:transitions}, the first to illustrate loops, and the second to illustrate \ttt{if} and \ttt{try} statements. For each transition, we note the applied inference rule as a subscript. If the conclusion of $\mathrm{[rule_1]}$ is used as a premise for $\mathrm{[rule_2]}$, we denote it by $\frac{\mathrm{[rule_1]}}{\mathrm{[rule_2]}}$.

\newcommand{\Ga}{
\tikz{
    \node (a) at (0,0) [draw,circle,line width=1pt,inner sep=.5ex] {};
    \node (b) at (.6,0) [draw,circle,line width=1pt,inner sep=.5ex] {};
    \node (c) at (1.2,0) [draw,circle,line width=1pt,inner sep=.5ex] {};
    \draw (a) edge[->] (b);
    \draw (b) edge[->] (c);
}}
\newcommand{\Gb}{
\tikz{
    \node (a) at (0,0) [draw,circle,line width=1pt,inner sep=.5ex] {};
    \node (b) at (.6,0) [draw,circle,line width=1pt,inner sep=.5ex] {};
    \draw (a) edge[->] (b);
}}
\newcommand{\Gbb}{
\tikz{
    \node (a) at (0,0) [draw,circle,line width=1pt,inner sep=.5ex] {};
    \node (b) at (.6,0) [draw,circle,line width=1pt,inner sep=.5ex] {};
    \draw (b) edge[->] (a);
}}
\newcommand{\Gc}{
\tikz{
    \node (a) at (0,0) [draw,circle,line width=1pt,inner sep=.5ex] {};
}}
\newcommand{\Gd}{
\tikz{
    \node (a) at (0,0) [draw,circle,line width=1pt,inner sep=.5ex] {};
    \node (b) at (.6,0) [draw,circle,line width=1pt,inner sep=.5ex] {};
    \node (c) at (1.2,0) [draw,circle,line width=1pt,inner sep=.5ex] {};
    \draw (b) edge[->] (a);
    \draw (b) edge[->] (c);
}}

\begin{figure}[p]
\centering

\begin{subfigure}[b]{\textwidth}
\renewcommand{\arraystretch}{1.4}
\begin{tabular}{p{.01cm} c l}
    \multicolumn{3}{l}{$\tuple{\ttt{r!},\,[\Ga]}$} \\
    
    $\to$&$_{\scriptstyle\mathrm{[alap_1]}}$ &
    $\tuple{\ttt{\tryte{r}{r!}{skip}},\,[\,\Ga\,]}$ \\
    
    $\to$&$_{\scriptstyle\mathrm{[try_1]}}$ &
    $\tuple{\text{TRY (\ttt{r}, \ttt{r!}, \ttt{skip})},\,[\,\Ga,\,\Ga\,]}$ \\
    
    $\to$&$_{\frac{\scriptstyle\mathrm{[call_1]}}{\scriptstyle\mathrm{[try_3]}}}$ &
    $\tuple{\ttt{r!}, \,[\Gb]}$ \\
    
    $\to$&$_{\scriptstyle\mathrm{[alap_1]}}$ &
    $\tuple{\ttt{\tryte{r}{r!}{skip}},\,[\,\Gb\,]}$ \\
    
    $\to$&$_{\scriptstyle\mathrm{[try_1]}}$ &
    $\tuple{\text{TRY (\ttt{r}, \ttt{r!}, \ttt{skip})},\,[\,\Gb,\,\Gb\,]}$ \\
    
    $\to$&$_{\frac{\scriptstyle\mathrm{[call_1]}}{\scriptstyle\mathrm{[try_3]}}}$ &
    $\tuple{\ttt{r!}, \,[\Gc]}$ \\
    
    $\to$&$_{\scriptstyle\mathrm{[alap_1]}}$ &
    $\tuple{\ttt{\tryte{r}{r!}{skip}},\,[\,\Gc\,]}$ \\
    
    $\to$&$_{\scriptstyle\mathrm{[try_1]}}$ &
    $\tuple{\text{TRY (\ttt{r}, \ttt{r!}, \ttt{skip})},\,[\,\Gc,\,\Gc\,]}$ \\
    
    $\to$&$_{\frac{\scriptstyle\mathrm{[call_2]}}{\scriptstyle\mathrm{[try_4]}}}$ &
    $\tuple{\ttt{skip}, \,[\Gc]}$ \\
    
    $\to$&$_{\scriptstyle\mathrm{[skip]}}$ &
    $[\Gc]$
\end{tabular}
\renewcommand{\arraystretch}{1}
\caption{Transition sequence of program \ttt{P} applied to graph \Ga.}
\label{fig:transition1}
\end{subfigure}

\bigskip
\begin{subfigure}[b]{\textwidth}
\renewcommand{\arraystretch}{1.4}
\setlength{\tabcolsep}{3pt}
\begin{tabular}{p{.01cm} c l}
    \multicolumn{3}{l}{$\tuple{\ttt{try(if(r1;r1) then(r1;r1))},\,[\,\Gd\,]}$} \\
    
    $\to$&$_{\scriptstyle\mathrm{[try_7]}}$ &
    $\tuple{\ttt{\tryte{(if (r1;r1) then (r1;r1))}{\skiptt}{\skiptt}},\,[\,\Gd\,]}$ \\
    
    $\to$&$_{\scriptstyle\mathrm{[try_1]}}$ &
    $\tuple{\text{TRY}(\ttt{if (r1;r1) then (r1;r1)},\,\skiptt,\,\skiptt), \,[\,\Gd,\,\Gd\,]}$ \\
    
    $\to$&$_{\frac{\scriptstyle\mathrm{[if_5]}}{\scriptstyle\mathrm{[try_2]}}}$ &
    $\tuple{\text{TRY}(\ttt{if (r1;r1) then (r1;r1) else skip},\,\skiptt,\,\skiptt), \,[\,\Gd,\,\Gd\,]}$  \\
    
    $\to$&$_{\frac{\scriptstyle\mathrm{[if_1]}}{\scriptstyle\mathrm{[try_2]}}}$ &
    $\tuple{\text{TRY}(\text{ITE}(\ttt{r1;r1},\,\ttt{r1;r1},\,\skiptt),\,\skiptt,\,\skiptt), \,[\,\Gd,\,\Gd,\,\Gd\,]}$ \\
    
    $\to$&$_{\frac{\scriptstyle\frac{\scriptstyle\frac{\scriptstyle\mathrm{[call_1]}}{\scriptstyle\mathrm{[seq_2]}}}{\scriptstyle\mathrm{[if_2]}}}{\scriptstyle\mathrm{[try_2]}}}$ &
    $\tuple{\text{TRY}(\text{ITE}(\ttt{r1},\,\ttt{r1;r1},\,\skiptt),\,\skiptt,\,\skiptt), \,[\,\Gb,\,\Gd,\,\Gd\,]}$ \\
    
    $\to$&$_{\frac{\scriptstyle\frac{\scriptstyle\mathrm{[call_1]}}{\scriptstyle\mathrm{[if_3]}}}{\scriptstyle\mathrm{[try_2]}}}$ &
    $\tuple{\text{TRY}(\ttt{r1;r1},\,\skiptt,\,\skiptt), \,[\,\Gd,\,\Gd\,]}$ \\
    
    $\to$&$_{\frac{\scriptstyle\mathrm{[call_1]}}{\scriptstyle\mathrm{[try_2]}}}$ &
    $\tuple{\text{TRY}(\ttt{r1},\,\skiptt,\,\skiptt), \,[\,\Gb,\,\Gd\,]}$ \\
    
    $\to$&$_{\frac{\scriptstyle\mathrm{[call_1]}}{\scriptstyle\mathrm{[try_3]}}}$ &
    $\tuple{\skiptt, \,[\,\Gc\,]}$ \\
    
    $\to$&$_{\mathrm{[skip]}}$ &
    $[\,\Gc\,]$
\end{tabular}
\renewcommand{\arraystretch}{1}
\caption{Transition sequence of program \ttt{P'} applied to graph \Gd.}
\label{fig:transition2}
\end{subfigure}

\caption{Examples of transition sequences}
\label{fig:transitions}
\end{figure}

\begin{example}
Consider the program \ttt{P=r!} and the rule 
\ttt{r} : $_1$\hspace{-.2em}\Gb \,$\Rightarrow$ \hspace{-.2em}$_1$\hspace{-.2em}\Gc. Let us examine a transition sequence of \ttt{P} applied to the graph \Ga, as seen in Figure \ref{fig:transition1}.

We start by applying $\mathrm{[alap_1]}$ which turns the loop into a \ttt{try} statement. Unlike in the previous semantics, we model a loop by trying to apply its body, and if it is successful, we call the loop again.

The inference rule $\mathrm{[try_1]}$ transforms the \ttt{try} statement into the auxiliary TRY construct, which advances the program in a small-step fashion, unlike the previous semantics. There is a similar ITE construct which models \ttt{if} statements. The top of the graph stack is duplicated since the changes made by the condition of the \ttt{try} may be discarded.

We then apply \ttt{r} to the current host graph (top of the stack) so $\mathrm{[call_1]}$ can be applied. This serves as a premise for $\mathrm{[try_3]}$, which ends the TRY statement, pops the second element of the stack, and moves on to the \ttt{then} part which is the original loop.

We repeat this process until \ttt{r} is no longer applicable to the host graph. At this point, $\mathrm{[call_2]}$ serves as the premise for $\mathrm{[try_4]}$ which exits the TRY statement. This time, the condition results in fail, so we move on to the \ttt{else} part which is \ttt{skip} and the loop terminates.

\medskip
Now consider program \ttt{P'=try(if (r1;r1) then (r1;r1))} and the rule \ttt{r1} : \mbox{$_1$\hspace{-.2em}\Gb $\,\Rightarrow$ \hspace{-.2em}$_1$\hspace{-.2em}\Gc}. A transition sequence of \ttt{P'} applied to host graph \Gd\, can be found in Figure \ref{fig:transition2}.

Since the \ttt{try} statement does not have a \ttt{then} or \ttt{else} part, we first apply $\mathrm{[try_7]}$, which adds \ttt{skip} as both the \ttt{then} and \ttt{else} parts.

The inference rule $\mathrm{[try_1]}$ turns the \ttt{try} statement into the auxiliary TRY statement and duplicates the top of the stack. For most of the remaining transition sequence, we apply $\mathrm{[try_2]}$ under various premises to advance the condition.

Since the \ttt{if} has no \ttt{else} part, $\mathrm{[if_5]}$ completes it with a \ttt{skip}. The \ttt{if} statement is then turned into the auxiliary ITE statement, duplicating the top of the stack once again.

The rule \ttt{r1} is applied to the host graph which advances the concatenation with $\mathrm{[seq_2]}$, the ITE with $\mathrm{[if_2]}$, and the TRY with $\mathrm{[try_2]}$. 
Calling \ttt{r1} a second time resolves the ITE, and the top of the stack is popped since changes made by the conditions of \ttt{if} statements are reversed.

We keep applying the condition of the TRY, until we resolve it with $\mathrm{[try_3]}$. This time the second graph on the stack is popped since changes made by the condition of a \ttt{try} that did not result in fail are preserved.
\qed
\end{example}

\section{Properties of the Semantics}
\label{sec:semantic-props}

In this section, we show that the semantics is non-blocking, i.e. if a transition sequence ends in an extended configuration, we can always apply an inference rule (Proposition \ref{prop:non-blocking}). Note that we can only guarantee the non-blocking property for extended configurations that are part of a transition sequence originating in a valid GP\,2 program. We call those \emph{reachable} extended configurations. This is reasonable because there can be no other types of configurations in a transition sequence modelling a GP\,2 program.

Furthermore, we will describe the outcomes of a transition sequence starting with a valid GP\,2 program (Proposition \ref{prop:outcomes}), and show that we have finite nondeterminism (Proposition \ref{prop:finite-nondet}), i.e. there are only finitely many one-step transitions starting from any configuration, and what it means for the semantic function.

Let us first look at a lemma that guarantees we can make a transition step from extended configurations that do not contain a \ttt{break}, which is the first step towards showing the non-blocking property.

\medskip
\begin{lemma}[Progress from Extended Configurations]\label{lem:one-step}
Let $\tuple{P,S}$ be an extended configuration. Then one of the following applies:
\begin{itemize}
    \item $\tuple{P,S} \to \tuple{P',S'}$ for some extended configuration $\tuple{P',S'}$.
    \item $\tuple{P,S} \to S'$ for some graph stack $S' \in \mathcal{S}$.
    \item $\tuple{P,S} \to \failrm$.
    \item $P$ is not a command sequence and contains a $\ttt{break}$.
\end{itemize}
\end{lemma}
\begin{proof}
We shall prove this lemma by going through what $P$ could be according to the syntax and the semantics.

\newcounter{cnum}
\setcounter{cnum}{1}

\emph{Case \thecnum\stepcounter{cnum}: $P$ is a rule set call.}
Then either $\text{top}(S) \dder_P G$ or $\text{top}(S) \not\dder_P$. So either $\mathrm{[call_1]}$ or $\mathrm{[call_2]}$ can be applied.

\emph{Case \thecnum\stepcounter{cnum}: $P$ is a loop.}
If $P$ is a loop, $\mathrm{[alap_1]}$ can be applied.

\emph{Case \thecnum\stepcounter{cnum}: $P$ is $\failrm$, $\skiptt$ or an \emph{\ttt{or}} statement.}
Then $\mrm{[fail]}$, $\mrm{[skip]}$, or $\mrm{[or_1]}$ can be applied respectively.

\emph{Case \thecnum\stepcounter{cnum}: $P$ is of the form $\ifte{P_1}{P_2}{P_3}$ or $\tryte{P_1}{P_2}{P_3}$.}
Then $\mathrm{[if_1]}$ or $\mathrm{[try_1]}$ can be applied. If any then-clause or else-clause is omitted as specified by the syntax, $\mathrm{[if_5]}$, $\mathrm{[try_5]}$, $\mathrm{[try_6]}$, or $\mathrm{[try_7]}$ can be applied.

\emph{Case \thecnum\stepcounter{cnum}: $P$ is of the form $\text{ITE}(P_1,\,P_2,\,P_3)$ or $\text{TRY}(P_1,\,P_2,\,P_3)$.}
If $P$ contains a \ttt{break}, the fourth point of the lemma is satisfied, as containing ITE or TRY statements makes $P$ not a command sequence. So for the remainder of this case, assume $P$ does not contain a \ttt{break}. If $P_1$ is a sequential composition, let $P_1=P_1';\,P_1''$ where $P_1'$ is not a sequential composition. Otherwise let $P_1=P_1'$. We shall show the lemma's statement by induction on how many ITE or TRY statements are nested in $P_1'$ via the first sequential component of the condition.
\begin{itemize}
    \item For the base case, assume $P_1'$ is not an ITE or TRY statement. Then $P_1'$ is not a sequential composition and covered by cases $1$ to $4$ ($P_1'$ cannot be \ttt{break} since $P$ contains no \ttt{break}).
    \item Now for the induction step, assume that $P_1'$ is an ITE or TRY statement. Then $P_1'$ does derive either a configuration $\tuple{P_1''',S'}$, a graph stack $S'$ or $\failrm$ by the induction hypothesis. Hence one of $\mathrm{[if_2]}$, $\mathrm{[if_3]}$, $\mathrm{[if_4]}$, $\mathrm{[try_2]}$, $\mathrm{[try_3]}$, or $\mathrm{[try_4]}$ can be applied to $\tuple{P,S}$.
\end{itemize}

\emph{Case \thecnum\stepcounter{cnum}: $P$ is a sequential composition.}
Then we can decompose $P$ into $P=P_1;P_2$ where $P_1$ is not a sequential composition. We can apply $\mathrm{[seq_1]}$, $\mathrm{[seq_2]}$, or $\mathrm{[seq_3]}$ since $\tuple{P_1,S} \to \tuple{P_1',S'}$, $\tuple{P_1,S} \to S'$, or $\tuple{P_1,S} \to \failrm$ respectively by cases 1 to 5.

\emph{Case \thecnum\stepcounter{cnum}: $P$ contains a \ttt{break}.}

If $P$ is not a command sequence, the final point of the lemma is satisfied. Otherwise, $P$ satisfies context conditions, meaning it must be enclosed within a loop, so either case $2$ or one of the other previous cases is applicable.

\end{proof}

\medskip
Lemma \ref{lem:one-step} has a case where the extended command sequence contains a \ttt{break}. This is because for a transition sequence not to get stuck on a \ttt{break}, we need to start with a command sequence where the \ttt{break} is within a loop, which we cannot guarantee if we consider a single transition step like in Lemma \ref{lem:one-step}. In order to deal with this case, we prove that we can construct a transition sequence that leads to a state with no \ttt{break} in the following lemma. However, we need to restrict it to extended configurations reachable from a valid GP\,2 program. We say that an extended configuration $C$ is \emph{reachable} if there is a configuration $\tuple{P,[G]}$ such that $\tuple{P,[G]} \to^* C$. This will still allow us to work towards non-blocking, since we only care about transition sequences that start with valid GP\,2 programs.

\medskip
\begin{lemma}[Removing the \ttt{break} Statement]\label{lem:break}
Let $\tuple{P,S}$ be an extended configuration that is reachable and non-terminal. Suppose that $P$ contains \ttt{break}. Then one of the following applies.
\begin{itemize}
\item There is an extended configuration $\tuple{P',S'}$ containing no \ttt{break} statement such that $\tuple{P,S} \to^* \tuple{P',S'}$.
\item There is a graph stack $S'$ such that $\tuple{P,S} \to^+ S'$.
\item $\tuple{P,S} \to^+ \failrm$
\end{itemize}
\end{lemma}
\begin{proof}
First assume that $\tuple{P,S}$ satisfies context conditions, i.e.\ the \ttt{break} is contained within a loop, and if the \ttt{break} is in the condition of an \ttt{if} or \ttt{try} statement, the enclosing loop must be in the same condition.

We will apply various inference rules to construct a transition sequence starting in $\tuple{P,S}$. Remember that whenever we apply such  an inference rule, it results in either a non-terminal extended configuration, a graph stack, or fail. If it results in a graph stack or fail, the second or third point of the lemma is satisfied. So at each step of the transition sequence we construct, we only need to consider the case where an inference rule results in a non-terminal extended command sequence.

If there are multiple loops with \ttt{break} statements, they are either in different sequential composition components, or nested. So let us show this lemma by induction on nesting and sequential composition.

As a base case, assume $P$ contains a single loop with a \ttt{break}, and want to show we can apply a sequence of inference rules that ultimately removes the \ttt{break}. So $P$ is of the form $Q_0 \ttt{;} Q_1 \ttt{!;} Q_2$, where $Q_1$ contains a \ttt{break}, and neither $Q_0$ nor $Q_2$ do. (What follows also applies if $P$ is of the form $Q_1 \ttt{!;} Q_2$, $Q_0 \ttt{;} Q_1 \ttt{!}$, or $Q_1 \ttt{!}$.) We can repeatedly apply Lemma \ref{lem:one-step} to transition to $Q_1 \ttt{!;} Q_2$. Then we apply $\mathrm{[alap_1]}$ followed by $\mathrm{[try_1]}$ to get $\text{TRY}(Q_1,Q_1!,\skiptt)$. We can then use Lemma \ref{lem:one-step} repeatedly as a premise for $\mathrm{[try_2]}$ until we get $\text{TRY}(Q_3\ttt{;}Q_4,Q_1!,\skiptt)$, where $Q_3$ contains \ttt{break} and is not a sequential composition. If $Q_3$ is a \texttt{break}, we can apply $\mathrm{[try_2]}$ under the premise of $\mathrm{[break]}$, followed by $\mathrm{[alap_2]}$ to get rid of the \texttt{break}. We know $Q_3$ cannot be a loop since we assumed $Q_1!$ is the enclosing loop of the \ttt{break}. So $Q_3$ is either an \ttt{or}, \ttt{if}, or \ttt{try} statement. If it is an \ttt{or} statement, we can apply $\mathrm{[or_1]}$ or $\mathrm{[or_2]}$ to either remove the \ttt{break} or lead to $\text{TRY}(\ttt{break}\ttt{;}Q_5,Q_1!,\skiptt)$. Similarly, if $Q_3$ is an \ttt{if} or \ttt{try} statement, the \ttt{break} must be in the \ttt{then} or \ttt{else} part due to context conditions, and we can use inference rules to either remove the \ttt{break} or lead to $\text{TRY}(\ttt{break}\ttt{;}Q_5,Q_1!,\skiptt)$. We can now apply $\mathrm{[try_2]}$ under the premise of $\mathrm{[break]}$ to get $\text{TRY}(\ttt{break},Q_1!,\skiptt)$. To this, we can apply $\mathrm{[alap_2]}$, which gets rid of the \ttt{break}.

For the induction step, let us first consider the case of nesting. Assume that $P$ is of the form $Q_0;(Q_1;Q_2;Q_3)!;$ $Q_4$, where $Q_2$ satisfies the lemma statement, and either $Q_1$ or $Q_2$ contain a single \ttt{break}. We can use the same arguments as in the base case in addition to $\mathrm{[seq_1]}$ under the premise of the induction hypothesis to get rid of the \ttt{break}.

Now consider sequential composition. As an induction step, assume that $P$ is of the form $Q_0;Q_1!;Q_2;$ $Q_3!;Q_4$, where one of $Q_1$ or $Q_3$ satisfies the lemma statement, and the other contains a single \ttt{break}. Again, we can use the arguments from the base case as well as the induction hypothesis in conjunction with $\mathrm{[seq_1]}$ to remove the \ttt{break}.

Finally, assume that $\tuple{P,S}$ does not satisfy context conditions, i.e.\ either there is a \texttt{break} without an enclosing loop, or there is a \texttt{break} in the condition of a branching statement whose enclosing loop is not within that condition. Since $\tuple{P,S}$ is reachable, the latter cannot be the case: transitions steps cannot separate a \texttt{break} from its enclosing loop in a way that they stop being within the same condition of a branching statement (loops can only be removed by inference rules, they cannot be ``moved''). So suppose there is a \ttt{break} without an enclosing loop. This must be because $\mathrm{[alap_1]}$ is applied earlier in the transition sequence, so it must be within the condition of a \ttt{try} or TRY. So we can use the same arguments as earlier in the proof, except that we need not argue that some of the inference rule, such as $\mathrm{[alap_1]}$ or $\mathrm{[try_1]}$ need to be applied.
\end{proof}

Now that we have Lemmata \ref{lem:one-step} and \ref{lem:break}, we can prove that the non-blocking property holds.

\medskip
\begin{proposition}[Non-Blocking Property]\label{prop:non-blocking}
Let $\tuple{P,S}$ be an extended configuration that is reachable and non-terminal. Then there is a transition step $\tuple{P,S} \to C$ for some extended configuration $C$.
\end{proposition}
\begin{proof}
If $P$ does not contain a \ttt{break}, this proposition follows from Lemma \ref{lem:one-step}. Otherwise, it follows from Lemma \ref{lem:break}.
\end{proof}

Let us now introduce a lemma that makes various statements about the size of host graph stacks in order to ensure that the inference rules are well-defined. Since we defined stacks to be nonempty, we want to make sure that if a transition sequence starts with a nonempty stack, it cannot lead to an empty stack, which the following lemma shows. Furthermore, when a transition sequence terminates in a graph stack, we want that stack to only contain one host graph.

For this lemma, we want to start from a valid GP\,2 program, not extended command sequences in general (since they may contain auxiliary constructs like ITE and TRY). So we consider configurations in $\text{CommandSeq} \times \mathcal{S}$. These follow the context conditions on where the \ttt{break} statement can appear as specified in \cite{Bak15a}.

\medskip
\begin{lemma}[Stack Size]\label{lem:stack}
Let $\tuple{P,[G]}$ be a configuration in $\text{CommandSeq} \times \mathcal{S}$.
\renewcommand{\labelenumi}{(\alph{enumi})}
\begin{enumerate}
    \item If $\tuple{P,[G]} \to^* \tuple{P',S}$, where $\tuple{P',S}$ is an extended configuration, then $|S| \geq 1$.
     \item If $\tuple{P,[G]} \to^+ S$, where $S$ is a graph stack, then $|S|=1$.
\end{enumerate}
\end{lemma}
\begin{proof}
The statement in (a), is satisfied for zero transition steps. So let us examine the inference rules that contain push, pop, and pop2. The rule $\mathrm{[call_1]}$ contains both push and pop, but preserves the size of the stack. The rules that push a graph onto the stack are $\mathrm{[if_1]}$ and $\mathrm{[try_1]}$ which are exactly the rules that introduce an ITE or a TRY. The rules that pop a graph from the stack are $\mathrm{[alap_2]}$, $\mathrm{[if_3]}$, $\mathrm{[if_4]}$, $\mathrm{[try_3]}$, and $\mathrm{[try_4]}$. These are exactly rules that remove an ITE or TRY from the extended command sequence. Since $\tuple{P,[G]}$ contains no ITE or TRY statements and only one host graph, we have $|S|=\#(P')+1$, where $\#$ counts the combined number of ITE and TRY statements in an extended command sequence. Since $|S|=\#(P')+1$, we have $|S| \geq 1$.

Now in case (b), we can break down the transition sequence into $\tuple{P,[G]} \to^* \tuple{P',S'} \to S$. Like in the proof of (a), the formula $|S'|=\#(P')+1$ applies. Let us examine which inference rules can be applied in the final step of the transition. It can only be either $\mathrm{[skip]}$, $\mathrm{[call_1]}$, or $\mathrm{[alap_2]}$. To apply $\mathrm{[skip]}$, $P'$ must be \ttt{skip} and $\#(\ttt{skip})=0$, so $|S|=|S'|=1$. To apply $\mathrm{[call_1]}$, $P'$ must be rule set call, and hence cannot contain ITE or TRY, so $|S|=|S'|=1$. To apply $\mathrm{[alap_2]}$, $P'$ must be of the form $\text{TRY}(\ttt{break},P''!,\ttt{skip})$, where $P''$ is an extended command sequence. We know $P''$ cannot contain an ITE or TRY statement because they can only be nested in their first argument. Indeed, if an extended command sequence already starts with an ITE or TRY, no inference rule allows for said ITE or TRY statement to be nested within another one. So the only way to nest statements is via the rule $\mathrm{[try_2]}$, which modifies the first argument. But the first argument of $P'$ is \ttt{break}, which contains no ITE or TRY statements. So $\#(P')=1$ and $|S'|=2$. Since we apply $\mathrm{[alap_2]}$, we have $S=\text{pop2}(S')$, so $|S|=|S'|-1=1$.
\end{proof}

We also want to make sure that if we call pop2 on a stack to pop its second element, the stack does indeed contain at least two elements. More precisely, under the premise of Lemma \ref{lem:stack}, if $\tuple{P,[G]} \to^+ \tuple{P',\text{pop2}(S)}$ (an extended configuration) or $\tuple{P,[G]} \to^+ \text{pop2}(S)$ (a graph stack), then $|S| \geq 2$. This follows directly from Lemma \ref{lem:stack} since $|\text{pop2}(S)|=|S|-1$.

\medskip
Let us now use Lemmata \ref{lem:one-step}, \ref{lem:break}, and \ref{lem:stack} is to describe what the possible outcomes of a transition sequence starting in a valid GP\,2 program are.

\medskip
\begin{theorem}[Outcomes of Transition Sequences]\label{prop:outcomes}
Let $\tuple{P,[G]}$ be a configuration. Then one of the following applies:
\begin{itemize}
    \item There is an infinite transition sequence $\tuple{P,[G]} \to \tuple{P_1,S_1} \to \tuple{P_2,S_2} \to \dots$ where $\tuple{P_i,S_i}$ is an extended configuration for all $i \geq 1$.
    \item $\tuple{P,[G]} \to^+ [G']$ for some host graph $G'$.
    \item $\tuple{P,[G]} \to^+ \text{fail}$.
\end{itemize}
\end{theorem}
\begin{proof}
Lemma \ref{lem:stack} guarantees that if a transition sequence starts in $\tuple{P,[G]}$ and ends in a stack, that stack only contains one graph. So for this proposition, it is enough to show that transition sequences end in a stack in the relevant cases.

In order to get rid of a potential \ttt{break} statement in $P$, we can apply Lemma \ref{lem:break} to $\tuple{P,[G]}$. If we get a graph stack or fail, we fulfil the second or third case of this proposition. Otherwise we get an extended configuration $\tuple{P',S}$ that contains no \ttt{break}.

Since there is now no \ttt{break} in either $\tuple{P,[G]}$ or $\tuple{P',S}$, we can apply the first, second, and third cases of Lemma \ref{lem:one-step} either indefinitely to get an infinite transition sequence, or until we get a graph stack or fail.
\end{proof}

Now that we know the possible outcomes of a transition sequence, we can define the \emph{semantic function} $\llbracket\_\rrbracket \,:\, \text{CommandSeq} \to (\mathcal{G} \to 2^{\mathcal{G}^{\oplus}})$, where $\mathcal{G}$ is the set of host graphs, $[\mathcal{G}]$ the set of stacks consisting of exactly one host graph (which we can identify with single host graphs), and $\mathcal{G}^{\oplus}=[\mathcal{G}] \,\cup\, \{\text{fail},\bot\}$. The symbol $\bot$ is used to represent an infinite transition sequence, i.e. divergence. The function is defined as
$$\llbracket P \rrbracket G = \{X \in [\mathcal{G}] \cup \{\text{fail}\} \,|\, \tuple{P,G} \to^+ X\} \,\cup\, \{\bot \,|\, P \text{ can diverge from } G\}\text{.}$$

    This functions differs from the one presented in \cite{Plump17a} and Section \ref{sec:gp2} since $\bot$ is only used when $P$ diverges, because we know by Proposition \ref{prop:non-blocking} that $P$ cannot get stuck. We will show in Lemma \ref{lem:infinite-new} that every infinite or stuck transition sequence in the previous semantics corresponds to an infinite transition sequence in the new semantics.

\bigskip
Let us now examine the property of \emph{finite nondeterminism} as specified by Apt in Section 4.1 of \cite{Apt84}, i.e. the set of elements reachable from a configuration in one transition step is finite. A related concept is \emph{bounded nondeterminism}, where the cardinality of the aforementioned set is finite and depends on the program only (and not on the size of the current state). An example for a language with bounded nondeterminsim is Dijkstra's language of guarded commands \cite{Reynolds98a}. Many references \cite{Dijkstra97,FokkinkVu03,Reynolds98a,SondergaardSestoft92} equate the concepts of finite and bounded nondeterminism and call it ``bounded nondeterminism''. However, rule-based languages generally have unbounded nondeterminism because they rely on nondeterministic rule matching. This also applies to GP\,2 which the following example illustrates.

\begin{example}
Consider the rule 
\ttt{r} : $_1$\hspace{-.2em}\Gbb \,$\Rightarrow$ \hspace{-.2em}$_1$\hspace{-.2em}\Gc\, and the \emph{comb graph} $G_4$ as shown in Figure \ref{fig:comb-graph}.
\begin{figure}[h]

\centering

\begin{tikzpicture}
    \node (a) at (0,0) [draw,circle,line width=1pt,inner sep=.7ex] {};
    \node (b) at (1,0) [draw,circle,line width=1pt,inner sep=.7ex] {};
    \node (c) at (2,0) [draw,circle,line width=1pt,inner sep=.7ex] {};
    \node (d) at (3,0) [draw,circle,line width=1pt,inner sep=.7ex] {};
    
    \node (e) at (0,1) [draw,circle,line width=1pt,inner sep=.7ex] {};
    \node (f) at (1,1) [draw,circle,line width=1pt,inner sep=.7ex] {};
    \node (g) at (2,1) [draw,circle,line width=1pt,inner sep=.7ex] {};
    \node (h) at (3,1) [draw,circle,line width=1pt,inner sep=.7ex] {};
    
    \draw (a) edge[->] (b);
    \draw (b) edge[->] (c);
    \draw (c) edge[->] (d);
    
    \draw (e) edge[->] (a);
    \draw (f) edge[->] (b);
    \draw (g) edge[->] (c);
    \draw (h) edge[->] (d);
\end{tikzpicture}

\caption{The comb graph $G_4$}
\label{fig:comb-graph}
\end{figure}
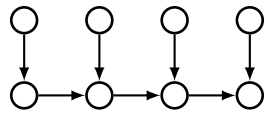
There are four possible matches for the left-hand side of rule \ttt{r} in graph $G_4$, so applying the rule can result in four different non-isomorphic graphs, which is a finite amount. When applying \ttt{r} to comb graph $G_k$, we get $k$ non-isomorphic graphs, which depends on the size of the host graph and hence is not bounded.
\qed
\end{example}

We now show that GP\,2 does have finite nondeterminism.

\medskip
\begin{proposition}[Finite Nondeterminism]\label{prop:finite-nondet}
Let $\gamma \in \text{ExtComSeq} \times \mathcal{S}$ be an extended configuration, and $T_{\gamma} = \{\gamma' \,|\, \gamma \to \gamma' \in \text{ExtCommSeq} \times \mathcal{S} \}$. Then $|T_{\gamma}|$ is finite.
\end{proposition}
\begin{proof}
The only inference rules that cause nondeterminism are $\mathrm{[or_1]}$, $\mathrm{[or_2]}$, and $\mathrm{[call_1]}$. If the rules $\mathrm{[or_1]}$ and $\mathrm{[or_2]}$ are applicable to $\gamma$ then there are exactly two configurations reachable from $\gamma$. In $\mathrm{[call_1]}$, the nondeterminism comes from several GP\,2 rules being called non-deterministically as part of a rule set, as well as from all the ways these rules can be matched in the host graph. Since rule sets and host graphs are finite, the number of configurations reachable from $\gamma$ in one step via the inference rule $\mathrm{[call_1]}$ is finite as well.
\end{proof}

Reynolds \cite{Reynolds98a} defines this kind of nondeterminism using the semantic function instead of the set of configurations reachable in one step. The following corollary shows that this semantics fulfils that definition as well.

\begin{corollary}\label{cor:reynolds-nondet}
Let $P \in \text{CommandSeq}$ and $G \in \mathcal{G}$ such that $\llbracket P \rrbracket G$ is infinite. Then $\bot \in \llbracket P \rrbracket G$.
\end{corollary}
\begin{proof}
Let $\gamma_0 = \tuple{P,[G]}$. Then $T^*_{\gamma_0} = \{\gamma \,|\, \gamma_0 \to^* \gamma \in \text{ExtCommSeq} \times \mathcal{S} \}$ is infinite as well since it contains all elements of $\llbracket P \rrbracket G$ except perhaps fail or $\bot$. The set $T^*_{\gamma_0}$ can be seen as a tree whose nodes are configurations and whose edges are defined by transition relations. Since $T_{\gamma}$ is finite for all configurations $\gamma$ by Proposition \ref{prop:finite-nondet}, each node in the tree only has finitely many adjacent nodes. By K\"onig's Lemma \cite{Konig27}, the tree contains an infinite path. Since every node of the tree is reachable from the root $\gamma_0$, there is an infinite path starting from $\gamma_0$. By definition of the tree, this means there is an infinite transition sequence starting with $\gamma_0$. By definition of the semantic function, we can conclude that $\bot \in \llbracket P \rrbracket G$.
\end{proof}

\section{Comparison to the Previous Semantics}
\label{sec:comparison}

In this section, we show that the new semantics is a conservative extension of the previous one, i.e. their behaviour is equivalent on converging configurations, and if a configuration diverges in the previous semantics, it also diverges in the new one.

When we mention graph stacks in this section, we allow them to be empty. We use the notation $[G_1,G_2, \dots ,G_k,S]$ (where $G_i$ are graphs, $S$ is a graph stack, and $k>0$) to denote a stack whose top $k$ elements are $G_1,G_2,\dots,G_k$, and whose remaining elements are the elements of $S$.

\begin{lemma}[Simulating Finite Old Transition Sequences]
\label{lem:finite-new}
Let $P \in \text{CommandSeq}$, $G \in \mathcal{G}$, and $X \in \{\tuple{P',\,G'},\, G',\, \text{fail}\}$, where $P' \in \text{CommandSeq}$ and $G' \in \mathcal{G}$. If $\tuple{P,\,G} \rightsquigarrow^* X$, then, for any graph stack $S$, there is a transition sequence
\begin{itemize}
    \item $\tuple{P,\,[G,S]} \to^* \tuple{P',\,[G',S]}$ if $X=\tuple{P',\,G'}$.
    \item $\tuple{P,\,[G,S]} \to^* [G']$ if $X=G'$.
     \item $\tuple{P,\,[G,S]} \to^* \text{fail}$ if $X=\text{fail}$.
\end{itemize}
\end{lemma}
\begin{proof}
We shall prove this lemma by induction on the number of \ttt{if}, \ttt{try}, and \ttt{!} statements in $P$ combined.

\smallskip
If $P$ has no \ttt{if}, \ttt{try}, or \ttt{!} statements, none of the $\mathrm{[if]}$, $\mathrm{[try]}$, and $\mathrm{[alap]}$ inference rules are applicable. All other rules behave identically in both semantics when identifying the tops graph of the stacks in the new rules with the graphs in the previous rules. Hence the base case is satisfied.

\smallskip
As the induction hypothesis, assume the lemma holds for $P$ containing $k$ \ttt{if}, \ttt{try}, or \ttt{!} statements. Now consider the case where $P$ contains $k+1$ of them. Let $\tuple{P_1,\,G_1} \rightsquigarrow \tuple{P_2,\,G_2}$ be a derivation step of $\tuple{P,G} \rightsquigarrow^* X$ that uses an $\mathrm{[if]}$, $\mathrm{[try]}$, or $\mathrm{[alap]}$ rule (possibly as a premise for another rule such as $\mathrm{[seq_1]}$). If such a step does not exist, the lemma holds by the same argument used in the base case. Consider the $\mathrm{[if]}$, $\mathrm{[try]}$, or $\mathrm{[alap]}$ rule $\mathrm{[r]}$ that relates to the \ttt{if}, \ttt{try}, or \ttt{!} statement enclosing all others resolved in the same step. Then no rule where $\mathrm{[r]}$ is a premise (or the premise of a premise) is an $\mathrm{[if]}$, $\mathrm{[try]}$, or $\mathrm{[alap]}$ rule. Since those are identical in both semantics, we only need to show that the part of $P$ resolved by $\mathrm{[r]}$ is resolved in a way that fulfils the lemma statement.

\smallskip
If $\mathrm{[r]}=\mathrm{[if_1]}$, we have $\tuple{\ifte{C}{P_3}{Q},\, G_3} \rightsquigarrow \tuple{P_3,\, G_3}$ under the premise of $\tuple{C,\, G_3} \rightsquigarrow^+ H$. Since $P$ contains $k+1$ \ttt{if}, \ttt{try}, or \ttt{!} statements, $C$ contains at most $k$ of them. So by the induction hypothesis, there is a transition sequence $\tuple{C,\, [G_3,S]} \to^* [H,S]$ (for any graph stack $S$). We can decompose this transition sequence into $\tuple{C,\, [G_3,S]} \to^l \tuple{C_4,\,S_4} \to [H,S]$ where $l\geq0$, and $S_4$ is a graph stack. This fulfills the premise of $\mathrm{[if_2]}$ $l$ times, and then the premise of $\mathrm{[if_3]}$ once. So for any graph stack $S'$, the premises are fulfilled by choosing $S=[G_3,S']$, and we have

\smallskip \noindent \begin{tabular}{r l}
$\tuple{\ifte{C}{P_3}{Q},\, [G_3,S']}$  
   &
   $\to_{\mathrm{[if_1]}} \tuple{\text{ITE}(C,P_3,Q),\,[G_3,G_3,S']} \to^l_{\mathrm{[if_2]}} \tuple{\text{ITE}(C_4,P_3,Q),\,S_4}$
   \\
     & 
   $\to_{\mathrm{[if_3]}} \tuple{P_2,[G_3,S']}\text{.}$
\end{tabular}

\smallskip
If $\mathrm{[r]}=\mathrm{[if_2]}$, we have $\tuple{\ifte{C}{P_3}{Q},\, G_3} \rightsquigarrow \tuple{Q,\, G_3}$ under the premise of $\tuple{C,\, G_3} \rightsquigarrow^+ \failrm$. Again, we can use the induction hypothesis to conclude $\tuple{C,\, [G_3,S]} \to^* \failrm$ (for any graph stack $S$), which is a sequence of $l\geq0$ transitions $\tuple{C,\, [G_3,S]} \to^l \tuple{C_4,\,S_4}$ between extended configurations, followed by $1$ transition $\tuple{C_4,\,S_4} \to \failrm$. This fulfills the premise of $\mathrm{[if_2]}$ $l$ times, and then the premise of $\mathrm{[if_4]}$ once. So for any graph stack $S'$, the premises are fulfilled by choosing $S=[G_3,S']$, and we have

\setlength{\tabcolsep}{2pt}
\renewcommand{\arraystretch}{1.1}
\smallskip \noindent \begin{tabular}{r l}
$\tuple{\ifte{C}{P_3}{Q},\, [G_3,S']}$  
   &
   $\to_{\mathrm{[if_1]}} \tuple{\text{ITE}(C,\,P_3,\,Q),\,[G_3,G_3,S']} \to^l_{\mathrm{[if_2]}} \tuple{\text{ITE}(C_4,\,P_3,\,Q),\,S_4}$
   \\
     & 
   $\to_{\mathrm{[if_4]}} \tuple{Q,\,[G_3,S']}\text{.}$
\end{tabular}

\smallskip
If $\mathrm{[r]}=\mathrm{[try_1]}$, we can use the same arguments as when $\mathrm{[r]}=\mathrm{[if_1]}$ to get a transition sequence $\tuple{\tryte{C}{P_3}{Q},\, [G_3,S']} \to^+ \tuple{P_3,\,[G_4,S']}$ for any graph stack $S'$.

\smallskip
If $\mathrm{[r]}=\mathrm{[try_2]}$, we can use the same arguments as when $\mathrm{[r]}=\mathrm{[if_2]}$ to get a transition sequence $\tuple{\tryte{C}{P_3}{Q},\, [G_3,S']} \to^+ \tuple{Q,\,[G_3,S']}$ for any graph stack $S'$. 

\smallskip
If $\mathrm{[r]}=\mathrm{[alap_1]}$, we have $\tuple{P_3!,\, G_3} \rightsquigarrow \tuple{P_3!,\, H}$ under the premise of $\tuple{P_3,\, G_3} \rightsquigarrow^+ H$. Since $P$ contains $k+1$ \ttt{if}, \ttt{try}, or \ttt{!} statements, $P_3$ contains at most $k$ of them. So by the induction hypothesis, we can conclude $\tuple{P_3,\, [G_3,S]} \to^* [H,S]$  for any graph stack $S$. We can decompose this transition sequence into $\tuple{P_3,\, [G_3,S]} \to^l \tuple{P_4,\,S_4} \to [H,S]$, where $l \geq 0$, and $S_4$ is a graph stack. These derivations fulfil the premise of $\mathrm{[try_2]}$ $l$ times and then the premise of $\mathrm{[try_3]}$ once. So for any graph stack $S'$, the premises are fulfilled by choosing $S=[G_3,S']$, and we have

\smallskip \noindent \begin{tabular}{r l}
$\tuple{P_3!,\, [G_3,S']}$  
   &
   $\to_{\mathrm{[alap_1]}} \tuple{\tryte{P_3}{P!}{\skiptt},[G_3,S']} \to_{\mathrm{[try_1]}} \tuple{\text{TRY}(P_3,\,P_3!,\skiptt),\,[G_3,G_3,S']}$
   \\
     & 
   $\to^l_{\mathrm{[try_2]}} \tuple{\text{TRY}(P_4,P_3!,\skiptt),\,S_4} \to_{\mathrm{[try_3]}} \tuple{P_3!,\,[H,S']} \text{.}$
\end{tabular}

\smallskip
If $\mathrm{[r]}=\mathrm{[alap_2]}$, we have $\tuple{P_3!,\, G_3} \rightsquigarrow G_3$ under the premise of $\tuple{P_3,\, G_3} \rightsquigarrow^+ \failrm$. Again, we can use the induction hypothesis to conclude $\tuple{P_3,\, [G_3,S]} \to^* \failrm$ for any graph stack $S$. We can decompose this transition sequence into $\tuple{P_3,\, [G_3,S]} \to^l \tuple{P_4,\,[G_4,S]} \to \failrm$, where $l \geq 0$, and $G_4$ is a graph. Let us argue that the graph stack right before the fail is $[G_4,S]$. The stack ends in $S$ because $P_3$ is a command sequence and hence does not contain ITE or TRY statements, which are the only constructs that could pop and hence modify $S$. Any \ttt{if}, \ttt{try}, and \ttt{!} statements in $P_3$ must resolve before the configuration that leads to fail in one step because fail is not reachable in one step from \ttt{if}, \ttt{try}, \ttt{!}, ITE, and TRY statements. Since they all resolved, any push has a corresponding pop, meaning the number of graphs in the stack before $S$ remains $1$. The derivations fulfil the premise of $\mathrm{[try_2]}$ $l$ times and then the premise of $\mathrm{[try_4]}$ once. So for any graph stack $S'$, the premises are fulfilled by choosing $S=[G_3,S']$, and we have

\smallskip \noindent \begin{tabular}{r l}
$\tuple{P_3!,\, [G_3,S']}$  
   &
   $\to_{\mathrm{[alap_1]}} \tuple{\tryte{P_3}{P!}{\skiptt},\,[G_3,S']} \to_{\mathrm{[try_1]}} \tuple{\text{TRY}(P_3,\,P_3!,\skiptt),\,[G_3,G_3,S']}$
   \\
     & 
   $\to^l_{\mathrm{[try_2]}} \tuple{\text{TRY}(P_4,P_3!,\skiptt),\,[G_4,G_3,S']} \to_{\mathrm{[try_4]}} \tuple{\skiptt,\,[G_3,S']} \to_{\mathrm{[skip]}} [G_3,S'] \text{.}$
\end{tabular}

\smallskip
If $\mathrm{[r]}=\mathrm{[alap_3]}$, we have $\tuple{P_3!,\, G_3} \rightsquigarrow H$ under the premise of $\tuple{P_3,\, G_3} \rightsquigarrow^* \tuple{\ttt{break},H}$. Again, we can use the induction hypothesis to conclude $\tuple{P_3,\, [G_3,S]} \to^* \tuple{\ttt{break},[H,S]}$ (for any graph stack $S$). Let $l\geq 0$ be the number of steps in that transition sequence. These derivations fulfil the premise of $\mathrm{[try_2]}$ $l$ times. So for any graph stack $S'$, the premises are fulfilled by choosing $S=[G_3,S']$, and we have

\smallskip \noindent \begin{tabular}{r l}
$\tuple{P_3!,\, [G_3,S']}$  
   &
   $\to_{\mathrm{[alap_1]}} \tuple{\tryte{P_3}{P_3!}{\skiptt},\,[G_3,S']} \to_{\mathrm{[try_1]}} \tuple{\text{TRY}(P_3,\,P_3!,\skiptt),\,[G_3,G_3,S']}$
   \\
     & 
   $\to^l_{\mathrm{[try_2]}} \tuple{\text{TRY}(\ttt{break},P_3!,\skiptt),\,[H,G_3,S']} \to_{\mathrm{[alap_2]}} [H,S'] \text{.}$
\end{tabular}

\vspace{-.51cm}

\end{proof}

\begin{lemma}[Configurations That Get Stuck]
\label{lem:stuck-conf}
Let $\tuple{P,G}$ be a configuration to which no inference rule of the previous semantics is applicable. Then $P$ starts with an \ttt{if}, \ttt{try}, or \ttt{!} statement such that the condition (or the body in the case of \ttt{!}) diverges or gets stuck in the previous semantics.
\end{lemma}
\begin{proof}
For this proof, whenever we say a rule is applicable, we mean it is either applicable, or it can be used as a premise for a $\mathrm{[seq]}$ rule.

Let us first show that, if $P$ does not start with an \ttt{if}, \ttt{try}, or \ttt{!} statement, we can apply an inference rule to $\tuple{P,G}$. If $P$ starts with a rule set, that rule set is either applicable to the host graph or not, so either $\mathrm{[call_1]}$ or $\mathrm{[call_2]}$ can be applied. If $P$ starts with a \ttt{break}, \ttt{skip}, or \ttt{fail}, then $\mathrm{[break]}$, $\mathrm{[skip]}$, or $\mathrm{[fail]}$ can be applied respectively. If $P$ starts with an \ttt{or} statement, either $\mathrm{[or_1]}$ or $\mathrm{[or_2]}$ can be applied.

Now assume that $P$ starts with an \ttt{if}, \ttt{try}, or \ttt{!} statement with a condition or body $C$. If $\tuple{C,G}$ neither converges nor gets stuck, there is a transition sequence $\tuple{C,G} \rightsquigarrow \failrm$, or $\tuple{C,G} \rightsquigarrow H$ for some graph $H$. Hence one of $\mathrm{[if_1]}$, $\mathrm{[if_2]}$, $\mathrm{[try_1]}$, $\mathrm{[try_2]}$, $\mathrm{[alap_1]}$, or $\mathrm{[alap_2]}$ is applicable.
\end{proof}

\begin{lemma}[Loop-Free Command Sequences Do Not Get Stuck]
\label{lem:non-stuck-conf}
For every loop-free command sequence $P$ and graph $G$, no transition sequence starting with $\tuple{P,G}$ gets stuck in the previous semantics.
\end{lemma}
\begin{proof}
We prove this lemma by structural induction. For a base case, consider programs consisting of a single rule set $R$. Then on any graph $G$, $R$ is either applicable or not. So to $\tuple{R,G}$ we can apply either $\mathrm{[call_1]}$ or $\mathrm{[call_2]}$, leading to a graph or fail. They are both transition sequences that end in a terminal state, and hence do not get stuck.

Another base case is \ttt{skip} or \ttt{fail}, which always lead to a graph or fail in a single step, and hence cannot lead to stuck transition sequences.

The \ttt{break} statement cannot be in $P$ since context conditions require it to have an enclosing loop.

For the induction step, assume that every proper subprogram of $P$ cannot get stuck, and show that $P$ cannot get stuck either.

Assume $P=P_1;P_2$. By the induction hypothesis, for any graph $G$, no transition sequence starting with $\tuple{P_1,G}$ gets stuck, i.e. they can all be extended to either $\tuple{P_1,G} \rightsquigarrow^+ \failrm$ or $\tuple{P_1,G} \rightsquigarrow^+ H$ for some graph $H$. This fulfils the premise of $\mathrm{[seq_1]}$ some number of times, and then the premise of either $\mathrm{[seq_2]}$ or $\mathrm{[seq_3]}$ once. So each transition sequence staring with $\tuple{P,G}$ must be of the form $\tuple{P,G} \rightsquigarrow^+_{\mathrm{[seq]}} \failrm$ or $\tuple{P,G} \rightsquigarrow^+_{\mathrm{[seq]}} H$.

Assume $P=\ifte{C}{P_1}{P_2}$. Then by induction hypothesis, for any graph $G$, no transition sequence starting with $\tuple{C,G}$ can get stuck, i.e. they can all be extended to either $\tuple{C,G} \rightsquigarrow^+ \failrm$ or $\tuple{C,G} \rightsquigarrow^+ H$ for some graph $H$. This satisfies the premise of either $\mathrm{[if_1]}$ or $\mathrm{[if_2]}$. So each transition sequence staring with $\tuple{P,G}$ must be of the form $\tuple{P,G} \rightsquigarrow_{\mathrm{[if_1]}} \tuple{P_1,G}$ or $\tuple{P,G} \rightsquigarrow_{\mathrm{[if_2]}} \tuple{P_2,G}$. Any continuation of these sequences cannot get stuck because $P_1$ and $P_2$ satisfy the induction hypothesis.

The case $P=\tryte{C}{P_1}{P_2}$ is analogous to the previous one.

If $P$ is an \ttt{if} or \ttt{try} with omitted \ttt{then} or \ttt{else} clauses, one of $\mathrm{[if_3]}$, $\mathrm{[try_3]}$, $\mathrm{[try_4]}$, or $\mathrm{[try_5]}$ can be applied, and then the arguments used in the \ttt{if} and \ttt{try} cases can be applied.

Assume $P=P_1\ttt{ or }P_2$. Then for any graph $G$, any transition sequence starting with $\tuple{P,G}$ starts with either $\tuple{P,G} \rightsquigarrow_{\mathrm{[or_1]}} P_1$ or $\tuple{P,G} \rightsquigarrow_{\mathrm{[or_2]}} P_2$. Any continuation of these sequences cannot get stuck because $P_1$ and $P_2$ satisfy the induction hypothesis.
\end{proof}

\begin{lemma}[Non-Nested Loops Do Not Get Stuck]
\label{lem:non-stuck-loops}
For every loop-free command sequence $P$ and graph $G$, no transition sequence starting with $\tuple{P!,G}$ gets stuck in the previous semantics.
\end{lemma}
\begin{proof}
It is enough to show that either $\tuple{P,G}\rightsquigarrow^+ H$, $\tuple{P,G} \rightsquigarrow^+ \failrm$, or $\tuple{P,G} \rightsquigarrow^* \tuple{\ttt{break},H}$. Because then, one of $\mathrm{[alap_1]}$, $\mathrm{[alap_2]}$, or $\mathrm{[alap_3]}$ is applicable. And if $\mathrm{[alap_1]}$ was applicable, we get $\tuple{P!,G} \rightsquigarrow \tuple{P!,H}$, the same arguments can be used on $\tuple{P!,H}$, and hence on all its successors.

First of all, since $P$ contains no \ttt{!}, $\tuple{P,G}$ cannot diverge. If $P$ does not contain a \ttt{break}, $\tuple{P,G}$ cannot get stuck by Lemma \ref{lem:non-stuck-conf}. If $P$ does contain a \ttt{break}, that is never called, the arguments of the proof of Lemma \ref{lem:non-stuck-conf} still apply, and $\tuple{P,G}$ does not get stuck. If $P$ contains a \ttt{break} that is called, that means there is a transition sequence $\tuple{P,G} \rightsquigarrow^* \tuple{\ttt{break},H}$, or $\tuple{P,G} \rightsquigarrow^* \tuple{\ttt{break};Q,H}$, to which we can apply $\mathrm{[break]}$ to get the former.

So either $\tuple{P,G} \rightsquigarrow^* \tuple{\ttt{break},H}$, or $\tuple{P,G}$ neither diverges nor gets stuck, which means $\tuple{P,G}$ must resolve to either a graph or fail.
\end{proof}

\begin{lemma}[Simulating Old Transition Sequences That Are Infinite or Stuck]
\label{lem:infinite-new}
Assume there is an infinite transition sequence $\tuple{P,G} \rightsquigarrow \dots$, or a stuck transition sequence $\tuple{P,G} \rightsquigarrow^* \tuple{P',G'}$. Then for any graph stack $S$, there is an infinite transition sequence $\tuple{P,[G,S]} \to \dots$.
\end{lemma}
\begin{proof}
First assume there is an infinite transition sequence $\tuple{P,G} \rightsquigarrow \tuple{P_1,G_1} \rightsquigarrow \dots$. To each step $\tuple{P_i,G_i} \rightsquigarrow \tuple{P_{i+1},G_{i+1}}$ in that transition sequence, we can apply Lemma \ref{lem:finite-new} to get a transition sequence $\tuple{P_i,[G_i,S_i]} \to^* \tuple{P_{i+1},[G_{i+1},S_{i+1}]}$ for any graph stack $S_i$ and some $S_{i+1}$. We can concatenate these into an infinite transition sequence.

Now assume there is a stuck transition sequence $\tuple{P,G} \rightsquigarrow^* \tuple{P',G'}$. We can apply Lemma \ref{lem:finite-new} to each step in that sequence to get $\tuple{P,[G,S]} \to^* \tuple{P',[G',S]}$ for each graph stack $S$.

We claim that for any command sequence $P'$ such that $\tuple{P',G'}$ is stuck with respect to $\rightsquigarrow$, there is a diverging transition sequence $\tuple{P',[G',S]} \to \dots$ for any graph stack $S$. This is enough to prove the Lemma. Let us show this by induction on the combined number of \ttt{if}, \ttt{try}, and \ttt{!} statements in $P'$. 

If there are no such statements, then $\tuple{P',G'}$ cannot get stuck by Lemma \ref{lem:stuck-conf}, so there has to be at least one, and $P'$ starts with it. If that one statement is an \ttt{if} or a \ttt{try}, then $\tuple{P',G'}$ cannot get stuck by Lemma \ref{lem:non-stuck-conf}. If it is a \ttt{!} statement, $\tuple{P',G'}$ cannot get stuck by Lemma \ref{lem:non-stuck-loops}. So there have to be at least two \ttt{if}, \ttt{try}, or \ttt{!} statements.

So for our base case, assume $P'$ contains exactly two \ttt{if}, \ttt{try}, or \ttt{!} statements. By Lemma \ref{lem:non-stuck-conf}, $P'$ must contain a \ttt{!} statement. If the two statements are not nested, we can apply Lemmata \ref{lem:non-stuck-conf} and/or \ref{lem:non-stuck-loops} sequentially to conclude that $\tuple{P',G'}$ does not get stuck. So the two statements must be nested, and they must be the start of $P'$ by Lemma \ref{lem:stuck-conf}. Assume the ``inner'' statement is not a \ttt{!} statement. Then the ``outer'' statement must be the \ttt{!} statement. By Lemma \ref{lem:non-stuck-loops}, $\tuple{P',G'}$ does not get stuck. So the ``inner'' statement must be a \ttt{!} statement $Q!$. The loop $Q!$ cannot resolve to a graph or fail because then, the ``outer'' statement could be resolved and $\tuple{P',G'}$ would not be stuck. By Lemma \ref{lem:non-stuck-loops}, $Q!$ cannot get stuck either. So $Q!$ must diverge, and so must the condition or body $C$ of the starting statement of $P'$. So there is an infinite transition sequence $\tuple{C,G'} \rightsquigarrow \dots$, and hence by Lemma \ref{lem:finite-new} to each step in that transition sequence, we get $\tuple{C,[G',S]} \to \dots$ for any graph stack $S$. This serves as a premise for $\mathrm{[if_2]}$ or $\mathrm{[try_2]}$, which we can use to get an infinite transition sequence $\tuple{P',G'} \to \dots$.

Now for the induction step, assume we get infinite transition sequences for programs with at most $k$ \ttt{if}, \ttt{try}, and \ttt{!} statements. Assume $P'$ has $k+1$ such statements. By Lemma \ref{lem:stuck-conf}, since $\tuple{P',G'}$ is stuck, $P'$ must start with one of those statements. So we can apply either $\mathrm{[if_1]}$, $\mathrm{[try_1]}$, or $\mathrm{[alap_1]}$ followed by $\mathrm{[try_1]}$ to get $\tuple{P',[G',S']} \to^+ \tuple{P'',[G',G',S']}$ for any graph stack $S'$, where $P''$ starts with either $\text{ITE}(C,Q,Q')$ or $\text{TRY}(C,Q,Q')$. Now $C$ has at least one less \ttt{if}, \ttt{try}, or \ttt{!} statement than $P'$, so we can apply the induction hypothesis to get an infinite transition sequence $\tuple{C,[G',G',S']} \to \dots$, which can serve as premises for infinitely many applications of $\mathrm{[if_2]}$ or $\mathrm{[try_2]}$ (or these inference rules provide the premises for $\mathrm{[seq_1]}$). Hence we have an infinite transition sequence $\tuple{P',[G',S']} \to^+ \tuple{P'',[G',G',S']} \to \dots$.
\end{proof}

\begin{lemma}[Simulating Finite New Transition Sequences]
\label{lem:finite-prev}
Let $P \in \text{CommandSeq}$, $G \in \mathcal{G}$, $S$ a graph stack, and $X \in \{\tuple{P',\,[G',S']},\, [G',S'],\, \text{fail}\}$, where $P' \in \text{CommandSeq}$, $S'$ is a graph stack, and $G' \in \mathcal{G}$. If $\tuple{P,\,[G,S]} \to^* X$, then there is a transition sequence
\begin{itemize}
    \item $\tuple{P,\,G} \rightsquigarrow^* \tuple{P',\,G'}$ if $X=\tuple{P',\,[G',S']}$.
    \item $\tuple{P,\,G} \rightsquigarrow^* G'$ if $X=[G',S']$.
     \item $\tuple{P,\,G} \rightsquigarrow^* \text{fail}$ if $X=\text{fail}$.
\end{itemize}
\end{lemma}
\begin{proof}
Let us show this lemma by induction on the combined number of \ttt{if}, \ttt{try}, and \ttt{!} statements in $P$. If $P$ has no such statements, no step in $\tuple{P,\,[G,S]} \to^* X$ uses $\mathrm{[if]}$, $\mathrm{[try]}$, or $\mathrm{[alap]}$ rules. So no pushing or popping occurs and only the top of the stack is modified. The applied rules behave identically to those in the previous semantics when identifying the top of the stacks in the new rules with the host graphs in the previous rules. Hence the lemma is satisfied in the base case.

Now assume that the lemma holds for command sequences with $k$ \ttt{if}, \ttt{try}, and \ttt{!} statements, and assume $P$ has $k+1$ of them. By the previous paragraph, we can simulate transition steps that do not involve $\mathrm{[if]}$, $\mathrm{[try]}$, or $\mathrm{[alap]}$ rules. So let $\tuple{Q,[H,S']}$ be a configuration (we will handle the case where $\tuple{Q,[H,S']}$ is an extended configuration, but not a configuration later) in the sequence $\tuple{P,\,[G,S]} \to^* X$ that uses such a rule to get the next state. Let $\mathrm{[r]}$ be that rule. It can only be $\mathrm{[if_1]}$, $\mathrm{[try_1]}$, or $\mathrm{[alap_1]}$ since $\tuple{Q,[H,S']}$ is a configuration and hence does not contain ITE or TRY. 

If $\mathrm{[r]} = \mathrm{[if_1]}$, then we get $\tuple{Q,[H,S']} \to \tuple{\text{ITE}(C,Q_1,Q_2);Q_3,[H,H,S']} \to^+ \tuple{Q_4;Q_3,[H,S']}$, where $Q_1$, $Q_2$, and $Q_3$ are command sequences, and where $Q_4 \in \{Q_1,Q_2\}$ (since this is part of a transition sequence that ends in $X$, a configuration, graph, or fail, we know that the \ttt{if} eventually resolves). So either $\tuple{C,[H,H,S']} \to^* [H',H,S']$ or $\tuple{C,[H,H,S']} \to^* \failrm$. We can apply the induction hypothesis to $C$, which has at most $k$ \ttt{if}, \ttt{try}, and \ttt{!} statements, to get that either $\tuple{C,H} \rightsquigarrow^* H'$ or $\tuple{C,H} \rightsquigarrow^* \failrm$. We can use this as a premise for either $\mathrm{[if_1]}$ or $\mathrm{[if_2]}$ to get $\tuple{Q,H} \rightsquigarrow \tuple{Q_4;Q_3,H}$, where $Q_4 \in \{Q_1,Q_2\}$, which simulates the sequence outlined at the beginning of this case.

If $\mathrm{[r]} = \mathrm{[try_1]}$, then we get either $\tuple{Q,[H,S']} \to \tuple{\text{TRY}(C,Q_1,Q_2);Q_3,[H,H,S']} \to^+ \tuple{Q_1;Q_3,[H',S']}$ or $\tuple{Q,[H,S']} \to \tuple{\text{TRY}(C,Q_1,Q_2);Q_3,[H,H,S']} \to^+ \tuple{Q_2;Q_3,[H,S']}$, where $Q_1$, $Q_2$, and $Q_3$ are command sequences, and $H'$ some graph. We can use the same arguments as in the previous case to get a transition sequence $\tuple{Q,H} \rightsquigarrow \tuple{Q_1;Q_3,H'}$ or $\tuple{Q,H} \rightsquigarrow \tuple{Q_2;Q_3,H}$.

If $\mathrm{[r]} = \mathrm{[alap_1]}$, we have $Q=Q_1!;Q_2$ so either $\tuple{Q_1!;Q_2,[H,S']} \to^* \tuple{Q_1!;Q_2,[H',S']}$, $\tuple{Q_1!;Q_2,$ $[H,S']} \to^* \tuple{Q_2,[H,S']}$, or $\tuple{Q_1!;Q_2,[H,S']} \to^* \tuple{\text{TRY}(\ttt{break},Q_1!,\skiptt);Q_2,[H',H,S']} \to   \tuple{Q_2,[H',$\\ $S']}$. We can conclude that either $\tuple{Q_1,[H,S']} \to^* [H',S']$, $\tuple{Q_1,[H,S']} \to^* \failrm$, or $\tuple{Q_1,[H,S']} \to^* \tuple{\ttt{break},$ $[H',S']}$. By induction hypothesis, we get that either $\tuple{Q_1,H} \rightsquigarrow^* H'$, $\tuple{Q_1,H} \rightsquigarrow^* \failrm$, or $\tuple{Q_1,H} \rightsquigarrow^* \tuple{\ttt{break},H'}$. These can be used as premises of $\mathrm{[alap_1]}$, $\mathrm{[alap_2]}$, or $\mathrm{[alap_3]}$ to conclude that either $\tuple{Q_1!;Q_2,$ $H} \rightsquigarrow \tuple{Q_1!;Q_2,H'}$, $\tuple{Q_1!;Q_2,H} \rightsquigarrow \tuple{Q_2,H}$, or $\tuple{Q_1!;Q_2,H} \rightsquigarrow \tuple{Q_2,H'}$.

Finally, if $\tuple{Q,[H,S']}$ is an extended configuration, but not a configuration, $Q$ must contain an ITE, TRY, or a \ttt{break} that does not satisfy the context conditions. All of these must originate in an \ttt{if}, \ttt{try}, or \ttt{!} statement, and are hence covered by the previous part of the proof.
\end{proof}

Note that the first point of Lemma \ref{lem:finite-prev} only applies to transition sequences between command sequences (they do not contain TRY or ITE constructs). So diverging transition sequences in the new semantics where all command sequences contain ITE or TRY after some step cannot be simulated with diverging transition sequences in the previous semantics.

\begin{theorem}
Let $P \in \text{CommandSeq}$ and $G \in \mathcal{G}$. Then
\renewcommand{\labelenumi}{(\alph{enumi})}
\begin{enumerate}
    \item $[P]G \subseteq \llbracket P \rrbracket G$\/ and
    \item $[P]G \setminus \{\bot\} = \llbracket P \rrbracket G \setminus \{\bot\}$.
\end{enumerate}
\end{theorem}
\begin{proof}
Lemma \ref{lem:finite-new} guarantees that (a) holds for graphs and fail, and Lemma \ref{lem:infinite-new} guarantees that (a) holds for $\bot$. Furthermore, (b) follows from (a) and Lemma \ref{lem:finite-prev}.
\end{proof}

\section{Conclusion}
\label{sec:conclusion}

We have introduced a new operational semantics for the graph programming language GP\,2. Unlike the previous semantics, this one is entirely small-step and non-blocking. As a consequence, it accurately models computations which intuitively should diverge (and do so in the implementation). In particular, the new semantic function correctly lists $\bot$ as an outcome when there is a computation in which the condition of a branching statement or the body of a loop diverges. We also obtain finite nondeterminism, meaning that for every configuration there are only finitely many choices for the next transition step.

Furthermore, we have shown that the new semantic function is an extension of the previous one, and that they are equivalent excluding divergence.

In future work, the new semantics should serve as a solid underpinning for setting up a time and space complexity theory for GP\,2. Its small-step nature is crucial to defining atomic computation steps. Such a theory could possibly be automated akin to the resource analysis in \cite{Moser20}.

Another aspect of the GP\,2 semantics is that it is orthogonal to the definition of the transformation rules. The inference rules that depend on the domain of graph transformation only need the definition of a rule application ($\mathrm{[call_1]}$) and the information when such an application fails ($\mathrm{[call_2]}$). Hence this semantics could be used as a foundation for GP\,2-like programming languages over other rule-based domains, such as string rewriting \cite{BookOtto93} or term rewriting \cite{BaaderNipkow98a}.

\bigskip \noindent\textbf{Acknowledgements.} We are grateful to the anonymous referees whose comments helped to improve the presentation of this paper.

\bibliographystyle{eptcs}
\bibliography{main}


\end{document}